\documentclass[a4paper,twocolumn]{article}
\usepackage[margin=3em, includeheadfoot]{geometry}
\usepackage{authblk}
\usepackage{cite}
\usepackage{amsfonts}
\usepackage{textcomp}
\usepackage[dvipdfmx]{graphicx}
\usepackage{amsmath,amssymb,bm}
\usepackage{braket}
\usepackage{color}
\usepackage{listings}
\usepackage{multirow}
\usepackage{adjustbox}
\usepackage{amsthm}
\usepackage{quantikz}
\usepackage{tikz}
\theoremstyle{definition}
\newtheorem{definition}{Definition}
\newtheorem{theorem}{Theorem}
\DeclareMathOperator{\qubits}{\#\text{qubits}}

\newcommand{\numberofonlyoptimizeradd}{1,015}

\begin{document}

% ===start: IEEE===
%\history{Date of publication xxxx 00, 0000, date of current version xxxx 00, 0000.}
%\doi{10.1109/TQE.2020.DOI}
% ===end: IEEE===

%\preprint{}

% ===start: arXiv===
\author[1]{Kaito Kishi}
\author[2]{Junpei Yamaguchi}
\author[2]{Tetsuya Izu}
\author[3]{Noboru Kunihiro}
\affil[1]{\textit{Quantum Laboratory, Fujitsu Research, Fujitsu Limited, 4-1-1 Kawasaki, Kanagawa 211-8588, Japan}}
\affil[2]{\textit{Data \& Security Laboratory, Fujitsu Research, Fujitsu Limited, 4-1-1 Kawasaki, Kanagawa 211-8588, Japan}}
\affil[3]{\textit{Institute of Systems and Information Engineering, University of Tsukuba, 1-1-1 Tennodai, Tsukuba, Ibaraki 305-8573, Japan}}
\date{}
% ===end: arXiv===

%\title{Estimation of Shor's DLP Circuit for 2048-bit Modulos based on Quantum Simulator}
\title{Simulation of Shor algorithm for discrete logarithm problems with comprehensive pairs of modulo $p$ and order $q$}

\renewcommand{\abstractname}{\vspace{-\baselineskip}}
\twocolumn[
    \begin{@twocolumnfalse}
        \maketitle
\vspace*{-3em}
% ===end: arXiv===
\begin{abstract}
%暗号分野でよく用いられる素体上の離散対数問題は古典計算機では解くための多項式時間アルゴリズムが見つかっていない。
The discrete logarithm problem (DLP) over finite fields, commonly used in classical cryptography, has no known polynomial-time algorithm on classical computers.
%一方、量子計算機による多項式時間アルゴリズムがShorにより与えられており、その理論的な解析もEkeråらによりおこなわれてきている。
However, Shor has provided its polynomial-time algorithm on quantum computers.
%しかし、一般の標数pと位数qのペアについて動作する具体的な量子回路をシミュレートした例はほとんど存在しない。
Nevertheless, there are only few examples simulating quantum circuits that operate on general pairs of modulo $p$ and order $q$.
%本研究では、そのような量子回路を具体的に構成し、
%量子シミュレータとPRIMEHPC FX700を利用して、
%32量子ビットまでで可能な1860通りすべての$p,q$のペアについて離散対数問題を解いた。
In this paper, we constructed such quantum circuits and solved DLPs for all 1,860 possible pairs of $p$ and $q$ up to 32 qubits using a quantum simulator with PRIMEHPC FX700.
%ここから、これまでヒューリスティックに解析されてきた成功確率の具体的な値を得て、
%その検証を実施した。
From this, we obtained and verified values of the success probabilities, which had previously been heuristically analyzed by Ekerå.
%%その結果、DLPを解くShorアルゴリズムの成功確率は、位数$q$による非対称的な形の周期性がみられることを発見した。
%As a result, we found that the success probability of Shor's algorithm for solving the DLP exhibits periodicity with an asymmetric waveform determined by the order $q$.
%%その結果、位数$q$による周期関数として知られる、DLPを解くShorアルゴリズムの成功確率は、具体的には$q$について左右に非対称的な波形となることを確認した。
%その結果、位数$q$による周期関数として知られる、DLPを解くShorアルゴリズムの成功確率の詳細な波形の形を明確にした。
As a result, the detailed waveform shape of the success probability of Shor's algorithm for solving the DLP, known as a periodic function of order $q$, was clarified.
%また、標数が2048ビットのときに必要な量子回路のサイズを、90量子ビットまでの量子回路を具体的に構成して、それらのリソース値から外挿して推定した。
%ここから、Schnorr型よりも必要な量子回路リソースが大きく、かつ標数pの小さな、ある種の上界となるpとqの組合せを与えた。
%また、より大きな$p,q$のペアを使った900通りの量子回路を生成し、そこから得られた回路規模を外挿して、$p=2048$~bitでの回路規模をsafe-prime群とSchnorr群で比較した。
Additionally, we generated \numberofonlyoptimizeradd{} quantum circuits for larger pairs of $p$ and $q$, extrapolated the circuit sizes obtained, and compared them for $p=2048$~bits between safe-prime groups and Schnorr groups.
%古典暗号ではsafe-prime群とSchnorr群の強度は標数$p$が等しければ同じだが、Shorの量子アルゴリズムを使うと後者は前者でのいくつの$p$のビット数にまで強度が落ちるのかを定量的に示した。
While in classical cryptography, the cipher strength of safe-prime groups and Schnorr groups is the same if $p$ is equal, we quantitatively demonstrated how much the strength of the latter decreases to the bit length of $p$ in the former when using Shor's quantum algorithm.
%特に、加算回路にripple carry adderを用いた場合、Shorの量子アルゴリズムのもとでp=2048 bitのSchnorr groupの暗号強度はp=1024 bitのsafe-prime groupとほぼ同等であることを示した。
In particular, it was experimentally and theoretically shown that when a ripple carry adder is used in the addition circuit, the cryptographic strength of a Schnorr group with $p=2048$~bits under Shor's algorithm is almost equivalent to that of a safe-prime group with $p=1024$~bits.
\end{abstract}
% ===start: arXiv===
\vspace*{2em}
\end{@twocolumnfalse}
]
% ===end: arXiv===

% ===start: APS===
%\keywords{}

% ===end: APS===
% ===start: IEEE===
%\begin{keywords}
%Cryptography,
%digital signatures,
%public key cryptography,
%quantum circuit,
%quantum cryptography,
%quantum simulation
%%Enter key words or phrases in alphabetical 
%%order, separated by commas. For a list of suggested keywords, send a blank 
%%email to keywords@ieee.org or visit \underline
%%{http://www.ieee.org/organizations/pubs/ani\_prod/keywrd98.txt}
%\end{keywords}
%
%\titlepgskip=-15pt
%\maketitle
% ===end: IEEE===

%\section{}
\section{Introduction}

%乗法的な有限群 $G$ のある元 $g$ が生成する巡回部分群を $\langle g \rangle$ とかくとき、巡回部分群のある元 $h\in\langle g\rangle$ に対して $h=g^s$ となる自然数 $0\leq s<q$ を求める問題を離散対数問題 (Discrete Logarithm Problem, DLP) と呼ぶ。
Let $G$ be a finite multiplicative group, and let $\langle g\rangle$ denote the cyclic subgroup generated by an element $g$ of $G$.
The problem of finding an integer $0\leq s<q$ such that $h=g^s$ for a given element $h\in\langle g\rangle$ is called the Discrete Logarithm Problem (DLP).
%離散対数問題の難しさは群によって変化するが、群 $G$ が標数 $p$ の有限体 $\mathbb{F}_p$ の既約剰余類群 $\mathbb{F}_p^\times$ であるとき、古典計算による多項式時間アルゴリズムは知られていないため、Diffie--Hellman (DH) 鍵共有~\cite{diffie-hellman}やDSA署名~\cite{dss}などの暗号アルゴリズムの安全性を保証するために利用されている。
The difficulty of the DLP varies depending on the group, but when $G$ is $\mathbb{F}_p^\times$, which is the multiplicative group of a finite field $\mathbb{F}_p$ of prime order $p$, no polynomial-time algorithm is known for classical computation.
This difficulty is utilized to ensure the security of cryptographic algorithms such as Diffie--Hellman (DH) key exchange~\cite{diffie-hellman} and Digital Signature Algorithm (DSA)~\cite{dss}.
%既約剰余類群 $\mathbb{F}_p^{\times}$ における離散対数問題を解く古典アルゴリズムとして
%数体篩法 \cite{nfs} が知られている.
The Number Field Sieve (NFS)~\cite{nfs} is a known classical algorithm for solving the DLP in the multiplicative group $\mathbb{F}_p^\times$.
%数体篩法を用いたこれまでの離散対数問題の求解実験では
%標数が240桁 (795ビット) の問題までが解けている (表 \ref{tab:DLPrecords} 参照) ことと, 
Experimental results using the NFS have solved DLP instances with modulos $p$ up to 240~digits (795~bits) (See Table~\ref{tab:DLPrecords}),
%数体篩法の計算量が $p$ のサイズに関する
and the computational complexity of the NFS is sub-exponential in the size of $p$,
%準指数関数 $O(\exp(((64/9)^{1/3}+o(1))(\ln p)^{1/3}(\ln\ln p)^{2/3}))$ になることから,
specifically $O(\exp(((64/9)^{1/3}+o(1))(\ln p)^{1/3}(\ln\ln p)^{2/3}))$.
%DH 鍵共有や DSA 署名を利用するには
%$p$ を2048ビット以上に設定することが推奨されている \cite{barker2018nistrecommendation}.
Therefore, it is recommended to set $p$ to be at least 2048~bits for DH key exchange and DSA~\cite{barker2018nistrecommendation}.
%一般的な巡回群における離散対数問題を解く古典アルゴリズムである Pohlig–Hellman 法 \cite{pohlig-hellman} 
%によると, 位数 $q$ は大きな素数であることが望ましい.
According to the Pohlig--Hellman algorithm~\cite{pohlig-hellman}, a classical algorithm for solving DLP in general cyclic groups, it is desirable for the order $q$ to be a large prime.
%既約剰余類群 $\mathbb{F}_p^{\times}$ においては常に $q|(p-1)$ が成立するので,
%$q=(p-1)/2$ が素数となるとき (つまり素数 $p$ が safe prime になるとき)に 
%Polig-Hellman 法に対して最も高い耐性を持つ.
In the multiplicative group $\mathbb{F}_p^\times$, $q$ always divides $p-1$, so when $q=(p-1)/2$ is prime (i.e., $p$ is a safe-prime), the group has the highest resistance to the Pohlig--Hellman algorithm.
%\textcolor{red}{（$q$は素数でありさえすればPohlig--Hellmanアルゴリズムは適用できないので、その大きさについては関係ない気がしますが、どうなのでしょうか？）}
%本稿では, $p,q$ が素数で $q=(p-1)/2$ を満たす組み合わせ $(p,q)$ をSafe-prime型と呼ぶ.
In this paper, we refer to pairs $(p,q)$ where $p$ and $q$ are primes and $q=(p-1)/2$ as safe-prime groups.
%DSA署名では素数 $q$ による剰余計算が必要になるため,
%$q$ は小さい方が望ましい.
For DSA signatures, it is desirable for $q$ to be small because modular arithmetic with the prime $q$ is required.
%位数 $q$ の一般的な巡回群における離散対数問題を解く古典アルゴリズムである Baby-step Giant-step 法
%\cite{baby-step} などの平方根法の計算量は $q$ のサイズに関する指数関数 $O(\sqrt{q})$ となるため,
%数体篩法に必要な計算時間と平方根法に必要な計算時間を比較することで,
%$p$ が2048ビットならば $q$ は224ビットまたは256ビット,
%$p$ が3072ビットならば $q$ は256ビットなどに設定することが推奨されている \cite{dss,dss-fips186-5}.
The computational complexity of square root algorithms such as the Baby-step Giant-step method~\cite{baby-step}, a classical algorithm for solving the DLP in general cyclic groups of order $q$, is exponential in the size of $q$, specifically $O(\sqrt{q})$. By comparing the computational time required for the NFS and square root algorithms, it is recommended to set $q$ to 224 or 256 bits when $p$ is 2048 bits, and $q$ to 256 bits when $p$ is 3072 bits~\cite{dss,dss-fips186-5}.
%本稿ではこのような関係を満たす組み合わせ $(p,q)$ をSchnorr型と呼ぶ.
In this paper, we refer to pairs $(p, q)$ that satisfy such relationships as Schnorr groups.
%重要なこととして、古典計算では$p$が一致し、$q$がSchnorr groupの$q$以上であれば、解くのに必要な計算量は変わらない。
Importantly, in classical computation, if $p$ matches and $q$ is at least the $q$ of the Schnorr group, the computational effort required to solve the DLP is the same.

%1994年に Shor は離散対数問題を多項式時間で求解可能な量子アルゴリズムを考案した \cite{shor}．
In 1994, Shor proposed a quantum algorithm that can solve the DLP in polynomial-time~\cite{shor-alt,shor}.
%このアルゴリズムは量子ビットの誤りを許容しないため,
%現在開発されている NISQ (Noisy Intermediate Scale Quantum) タイプの量子計算機において
%大規模な問題を解くことはできないが, 
%誤りのない FTQC (Fault Torelant Quantum Computer) タイプの量子計算機が実現した場合,
%離散対数問題を利用した暗号アルゴリズムが解読される危険性が高い.
This algorithm does not tolerate errors in qubits, so it cannot solve large-scale problems on currently developed NISQ (Noisy Intermediate Scale Quantum) computers.
However, if an FTQC (Fault-Tolerant Quantum Computer) is realized, there would be a high risk that cryptographic algorithms based on the DLP will be broken.

%Shorの量子アルゴリズムは素因数分解問題を多項式時間で解くことも可能であるため,
%%量子計算機実機・シミュレータでの素因数分解実験や
%2048ビットRSA型合成数を素因数分解するために必要なリソース量の見積もりが報告されている
Shor's quantum algorithm can also solve the factorization problem in polynomial-time,
and estimation of the resources required to factor a 2048~bit RSA composite number have been reported~\cite{Gidney2021howtofactorbit,Gouzien2021factoring,yamaguchi}.
%\cite{Gidney2021howtofactorbit}ではthe coset representation of modular integers, windowed arithmetic, oblivious carry runwaysを利用してtoffoli gatesをなるべく減らしたうえで、magic state factoryの時間を考慮した表面符号で実装した際の考察をしている。
In \cite{Gidney2021howtofactorbit}, the coset representation of modular integers, windowed arithmetic, and oblivious carry runways are used to minimize the number of Toffoli gates, and considerations are made for implementation using surface codes that account for the time of the magic state factory.
%そして、2048ビットRSA型合成数を素因数分解する際に必要なtoffoli gate数および計算時間を見積もっている。
This approach estimates the number of Toffoli gates and the computation time required to factorize a 2048~bit RSA composite number.
%似たアプローチとして、\cite{Gouzien2021factoring}では3次元色符号と量子メモリを利用して、\cite{Gidney2021howtofactorbit}と比べて実行時間は増えるものの、少ない量子ビット数での素因数分解実行時間を見積もっている。
Similarly, \cite{Gouzien2021factoring} estimates the factorization time using fewer qubits, although with increased time compared to \cite{Gidney2021howtofactorbit}, by using 3D color codes and quantum memory.
% とはいえ、このリソース推定値は実機の構成に強く依存し、特にnon-Cliffordゲートの実現方法はSTARやmagic state cultivationで大きく改善されてきており、現在は相対的にTやToffoli数に縛られない純粋なゲート数での見積もりも求められる。
Nevertheless, this resource estimate is highly dependent on the configuration of the actual quantum machine, and particularly the implementation of non-Clifford gates has been significantly improved by methods such as STAR~\cite{akahoshi2023partiallyftqc,toshio2024practicalquantumadvantage,akahoshi2024compilationoftrotter}, zero-level distillation~\cite{Itogawa2024zerolevel}, and magic state cultivation~\cite{gidney2024msc}.
%Currently, estimates based on 
%\textcolor{red}{(TODO: refer to STAR and MSC\cite{gidney2024msc} to do more than non-Clifford))}
%一方、ボトムアップ的にシミュレータ上で実際に実行可能な量子回路を構成して大規模にシミュレートし、そこで得られた具体的なゲート数から外挿することで、2048ビット合成数でのゲート数を見積もる方法も存在する\cite{yamaguchi,yamaguchi2023experimentsresource}。
On the other hand, a bottom-up approach involves constructing quantum circuits that can be practically executed on simulators and extrapolating from the obtained gate counts to estimate the gate count for a 2048~bit composite number~\cite{yamaguchi,yamaguchi2023experimentsresource}.

%一方、DLPを解くShorアルゴリズムについては、素因数分解よりも見積もりの解析が少ない。
There is less analysis on the estimation of solving the DLP using Shor's algorithm compared to the factorization.
%先述した\cite{Gidney2021howtofactorbit}では、DLPについてもsafe-prime groupとSchnorr groupそれぞれについて計算時間を見積もっている。
In the aforementioned \cite{Gidney2021howtofactorbit}, the computation time for DLP is estimated for both safe-prime groups and Schnorr groups.
%Shorのアルゴリズム\cite{shor-alt}やEkeråのアルゴリズム\cite{ekera2021quantumalgorithmsgeneraldlp}、Ekerå--Håstad\cite{ekera2017quantumalgorithmshortdlp}のアルゴリズムにより、後者の方が前者よりも短い時間で求解可能であると計算されている。
It shows that the latter can be solved in a shorter time than the former using Shor's algorithm~\cite{shor-alt}, Ekerå's algorithm~\cite{ekera2021quantumalgorithmsgeneraldlp}, and Ekerå--Håstad's algorithm~\cite{ekera2017quantumalgorithmshortdlp}.
%さらに、\cite{hhan2024quantumcomplexityfordlp}においては$\Omega(\log q)$ depth of group operation queries必要であると理論的に示している。
Furthermore, \cite{hhan2024quantumcomplexityfordlp} theoretically demonstrates that $\Omega(\log q)$ depth of group operation queries is required for solving the DLP.
%しかし、Shorのアルゴリズムにかかる時間で最も支配的な冪剰余計算を含めたときに、$p$と$q$が具体的にどのように時間計算量に関わるかは不明瞭である。
However, it remains unclear how $p$ and $q$ specifically affect the time complexity when including the most time-consuming part of Shor's algorithm, the modular exponentiation.
%また、成功確率のヒューリスティックな理論解析がEkeråによりなされている\cite{martin2019revisitingshor}が、厳密に正しくShorの量子アルゴリズムを動かしたときにその通りになるかは検証する必要がある。
Although heuristic theoretical analysis of the success probability has been conducted by Ekerå~\cite{martin2019revisitingshor}, it needs to be verified whether this holds true when Shor's quantum algorithm is executed precisely.
%実装例としては、Qiskitによるシミュレータ上での実験\cite{mandl2022implementationsshordlp}や、実機での実験が存在する\cite{aono}。
As for implementation examples, there are experiments on simulators using Qiskit~\cite{mandl2022implementationsshordlp} and on actual quantum computers~\cite{aono}.
%前者はゲート数や成功確率をプロットしているものの、それらが$p$と$q$のペアにどのように依存するかは考慮されていない。
The former plots the number of gates and success probability, but does not consider how these depend on the pair of $p$ and $q$.
%後者については、$2^x\equiv 1\pmod 3$ ($p=3,q=2$) を解く実験のみが成功しており、そこから大きな値での解析は難しい。
For the latter, only the experiment solving $2^x\equiv 1\pmod 3$ ($p=3,q=2$) has been successful, making it difficult to analyze for larger values.
%特に量子計算機がFTQCに近づいていくにつれ、より大きな$p$や$q$についてShorアルゴリズムは実行されるであろうから、その周辺での具体的な計算例を事前に集めておく必要がある。
Especially as quantum computers approach FTQC, Shor's algorithm will likely be executed for larger $p$ and $q$, so it is necessary to gather specific computational examples around this in advance.
%量子計算機がFTQCに近づいていくにつれ、まずは$p$や$q$がある程度小さいところからShorアルゴリズムは実行されるであろうから、その周辺での具体的な計算例を集めておく需要がある。
%他方、RegevによるShorアルゴリズムの多次元版\cite{regev2023efficientquantumfactoring}を離散対数問題にさらに拡張した研究が存在する\cite{ekera2024extendingregevsfactoringtodlp}が、これはゲート数が削減されるものの量子ビット数が増大して量子シミュレーションが困難になるため、今回は取り扱わない。
On the other hand, there is research extending Regev's multidimensional version of Shor's algorithm~\cite{regev2023efficientquantumfactoring} to the DLP~\cite{ekera2024extendingregevsfactoringtodlp},
but this will not be addressed here as it increases the number of qubits, making quantum simulation difficult despite the reduction in gate count.

%そこで、本研究ではDLPを解くShorアルゴリズムの量子回路を大規模にシミュレートした。
In this study, we comprehensively simulated the quantum circuits of Shor's algorithm for solving the DLP.
%加算回路にQ-ADD\cite{qadd}を用いたときに32量子ビットまでで可能なすべての$p$,$q$の組合せ1860通りの離散対数問題を量子シミュレータ上で解き、そのときの実際の成功確率を得た。
We solved DLPs on a quantum simulator for all 1,860 pairs of $p$ and $q$ possible up to 32~qubits using the Q-ADD~\cite{qadd} adder circuit,
and obtained the actual success probabilities.
%そこから、成功確率の分布や位数$q$の大きさによる変化を、Ekeråの解析\cite{martin2019revisitingshor}と比較した。
From this, we compared the distribution of success probabilities and the changes due to the size of the order $q$ with Ekerå's heuristic analysis~\cite{martin2019revisitingshor}.
%また、より大きな$p$,$q$の組合せ900通りを選び、加算回路にR-ADD\cite{radd}を用いたときの量子回路の最適化前後のゲート数や深さを得た。
Additionally, we selected \numberofonlyoptimizeradd{} pairs of larger $p$ and $q$, and obtained the number of gates and depth of the quantum circuits before and after optimization using the R-ADD~\cite{radd} adder circuit.
%そこから、\cite{yamaguchi,yamaguchi2023experimentsresource}のように2048ビットでの回路規模を外挿してボトムアップ的に推定した。
Using these results, we extrapolated the circuit scale for 2048~bit as in \cite{yamaguchi,yamaguchi2023experimentsresource} in a bottom-up manner.
%こうして、古典計算の枠組みではsafe-prime groupと同等の強度であるSchnorr groupが、量子計算機のもとではどの大きさの$p$のsafe-prime groupの強度にまで低下するのかを具体的に示した。
In this way, we specifically demonstrated how the cipher strength of the Schnorr group, which is equivalent to the safe-prime group in classical computation, decreases to the strength of a safe-prime group with a smaller $p$ under quantum computation.

In this paper, we use the following notations.
%自然数 $x>0$ に対し, 2を底とする対数を$\log x$，自然対数を$\ln x$と表す．
For a natural number $x>0$, the logarithm to the base 2 is denoted by $\log x$, and the natural logarithm is denoted by $\ln x$.
%また0以上の整数$x$に対し，そのサイズ (ビット長) を $|x|$, 
For a non-negative integer $x$, its size (bit length) is denoted by $|x|$,
%その小数点以下を四捨五入する関数を$\lceil x\rfloor$と表す．
and the function that rounds $x$ to the nearest integer is denoted by $\lceil x\rfloor$.

\begin{table*}[h]
    \begin{center}
    \caption{
      %既約剰余類群 $\mathbb{F}_p^{\times}$ における離散対数問題の求解記録の推移
      The progress of solving records for the DLP in the multiplicative group of a finite field $\mathbb{F}_p^{\times}$.
      %なお、strong primeとは、大きい素数$p$のうち$p-1$と$p+1$の最大の素因数がともに大きいものを指し、詳しくは\cite{rivest2001arestrongprimes}を参照。
      A strong-prime refers to a large prime number $p$ for which the largest prime factors of both $p-1$ and $p+1$ are also large; for more details, see \cite{rivest2001arestrongprimes}.
    }
    \label{tab:DLPrecords}
    \begin{tabular}{llll}
    \hline
      %求解日         & 標数 & \multicolumn{1}{c}{求解者}                     & \multicolumn{1}{c}{備考} \\
      Date solved & Modulo & \multicolumn{1}{c}{Solver} & \multicolumn{1}{c}{Remarks} \\
    \hline
      %2005年06月18日 & 130桁 & Joux, Lercier \cite{Jou05}                   & パターン 1 (strong prime) \\
      June 18, 2005 & 130 digits & Joux, Lercier \cite{Jou05} & strong-prime \\
      %2006年12月22日 & 135桁 & Dorofeev, Dygin, Matyukhin \cite{Dor06}      & パターン 1 (strong prime) \\
      December 22, 2006 & 135 digits & Dorofeev, Dygin, Matyukhin \cite{Dor06} & strong-prime \\
      %2007年02月05日 & 160桁 & Kleinjung \cite{Kle07}                       & パターン 1 (safe prime) \\
      February 5, 2007 & 160 digits & Kleinjung \cite{Kle07} & safe-prime \\
      %2014年06月11日 & 180桁 & Bouvier, Gaudry, Imbert et al. \cite{BGI+14} & パターン 1 (safe prime) \\
      June 11, 2014 & 180 digits & Bouvier, Gaudry, Imbert et al. \cite{BGI+14} & safe-prime \\
      %2016年06月16日 & 232桁 & Kleinjung, Diem, Lenstra et al. \cite{Kle16} & パターン 1 (safe prime) \\
      June 16, 2016 & 232 digits & Kleinjung, Diem, Lenstra et al. \cite{Kle16} & safe-prime \\
      %2019年02月02日 & 240桁 & Boudot et al. \cite{Tho19,BGG+20a,BGG+20b}   & パターン 1 (safe prime) \\
      February 2, 2019 & 240 digits & Boudot et al. \cite{Tho19,BGG+20a,BGG+20b} & safe-prime \\
    \hline
    \end{tabular}
    \end{center}
\end{table*}

\section{Shor's Algorithm for the DLP}

%離散対数問題を解くShorのアルゴリズムは大きく2段階に分けられる：
Shor's algorithm for solving the DLP can be divided into two main stages:
\begin{enumerate}
    %\item 冪剰余の計算および逆量子フーリエ変換，そして測定をおこなう量子計算パート
    \item The quantum computation part, which involves the calculation of modular exponentiation, the inverse quantum Fourier transform, and measurement.
    %\item 測定結果を利用して$s$を得る古典計算パート
    \item The classical computation part, which uses the measurement results to obtain $s$.
\end{enumerate}

%それぞれ順を追って記述する．
Each step will be described in sequence.

\subsection{The quantum computation part}

%量子ビットの集合を3つのレジスタに分け，それぞれの量子ビット数を$v_1,v_2,w\in\mathbb{N}$とし，次の状態をアダマールゲートにより準備する．
A set of qubits is divided into three registers, with the number of qubits in each register denoted as $v_1, v_2, w\in\mathbb{N}$, and the following state is prepared using the Hadamard gate:
\begin{align}
    \frac{1}{\sqrt{2^{v_1+v_2}}}\sum_{x_1=0}^{2^{v_1}-1}\sum_{x_2=0}^{2^{v_2}-1}\ket{x_1}\ket{x_2}\ket{0^w}.
\end{align}
Next, the X gate is applied to the least significant bit of the third register:
\begin{align}
    \frac{1}{\sqrt{2^{v_1+v_2}}}\sum_{x_1=0}^{2^{v_1}-1}\sum_{x_2=0}^{2^{v_2}-1}\ket{x_1}\ket{x_2}\ket{0\cdots 01}.
\end{align}
%次に冪乗剰余 $h^{x_1} g^{x_2} \bmod{p}$ を計算し,
%3つめのレジスタに保存する. 
Then, the modular exponentiation $h^{x_1} g^{x_2} \bmod{p}$ is computed and stored in the third register.
%ここで素数 $p$ による剰余計算が必要となるが,
%具体的な計算法としては R-ADD，GT-ADD，Q-ADD などが知られている \cite{kunihiro}.
The specific method for this computation is described in the later section Construction of Quantum Cirtuit.
%次に，冪剰余のオラクル計算を実行する．
%今回は可能な限り大きな$p$, $q$について量子シミュレーションができるよう，必要な補助量子ビット数を抑えられる手法~\cite{qadd}を用いた．
%一方，量子ビット数よりもゲート数を抑える必要のある場合は\cite{radd}で置き換えることができる\cite{kunihiro}．
%その計算を実行した結果が以下となる．
The result of this computation is as follows:
\begin{align}
    \frac{1}{\sqrt{2^{v_1+v_2}}}\sum_{x_1=0}^{2^{v_1}-1}\sum_{x_2=0}^{2^{v_2}-1}\ket{x_1}\ket{x_2}\ket{h^{x_1}g^{x_2}\bmod p}.
\end{align}
%これは以下と等価である\cite{qcqi}．
This is equivalent to the following~\cite{qcqi}:
\begin{align}
    &\frac{1}{\sqrt{2^{v_1+v_2}q}}\sum_{l=0}^{q-1}\sum_{x_1=0}^{2^{v_1}-1}\sum_{x_2=0}^{2^{v_2}-1}e^{2\pi i(slx_1+lx_2)/q} \nonumber\\
    &\quad\cdot\ket{x_1}\ket{x_2}\ket{\hat{f}(sl,l)} \\
    =&\frac{1}{\sqrt{2^{v_1+v_2}q}}\sum_{l=0}^{q-1}\left(\sum_{x_1=0}^{2^{v_1}-1}e^{2\pi ix_1(sl/q)}\ket{x_1}\right) \nonumber\\
    &\quad\cdot\left(\sum_{x_2=0}^{2^{v_2}-1}e^{2\pi ix_2(l/q)}\ket{x_2}\right)\ket{\hat{f}(sl,l)}
\end{align}
%ただし， $\ket{\hat{f}(l_1,l_2)}=\frac{1}{\sqrt{q}}\sum_{j=0}^{q-1}e^{-2\pi il_2j/q}\ket{g^j\bmod p}$とする．
where $\ket{\hat{f}(l_1,l_2)}=\frac{1}{\sqrt{q}}\sum_{j=0}^{q-1}e^{-2\pi il_2j/q}\ket{g^j\bmod p}$.
%逆量子フーリエ変換を適用すると次のようになる．
Applying the inverse quantum Fourier transform results in the following:
\begin{align}
    \frac{1}{\sqrt{q}}\sum_{l=0}^{q-1}\ket{\widetilde{sl/q}}\ket{\widetilde{l/q}}\ket{\hat{f}(sl,l)}.
\end{align}
%ここで，第1レジスタの$\widetilde{sl/q}$は，
Here, $\widetilde{sl/q}$ in the first register is the $v_1$-bit approximation of $\beta/q$ in
\begin{align}
    sl/q=\alpha+\beta/q\quad (\alpha,\beta\in\mathbb{Z}, 0<\beta<q) \label{eq:tilde-sl-q}
\end{align}
%のうち$\beta/q$の$v_1$ビット近似値であり，
%第2レジスタの$\widetilde{l/q}$は$l/q$の$v_2$ビット近似値である．
and $\widetilde{l/q}$ in the second register is the $v_2$-bit approximation of $l/q$.
%この第1，第2レジスタを測定して量子計算パートは終了となる．
Measuring the first and second registers concludes the quantum computation part.
%ここまでの手続きを量子回路で記述すると図\ref{fig:circuit-standard-dlp}のようになる。
The procedure up to this point is described by the quantum circuit shown in Figure~\ref{fig:circuit-standard-dlp}.

\begin{figure*}
\centering
\begin{tikzpicture}
\node[scale=0.6] {
\begin{quantikz}
    \lstick[5]{1st register} &[1em] \lstick{$\ket{+}$}     & \ctrl{10}                   & \qw                         & \cdots & \qw            & \qw            & \qw                         & \qw                         & \cdots & \qw            & \qw \\
                             & \lstick{$\ket{+}$}          & \qw                         & \ctrl{9}                    & \cdots & \qw            & \qw            & \qw                         & \qw                         & \cdots & \qw            & \qw \\
                             & \vdots \\
                             & \lstick{$\ket{+}$}          & \qw                         & \qw                         & \cdots & \ctrl{7}       & \qw            & \qw                         & \qw                         & \cdots & \qw            & \qw \\
                             & \lstick{$\ket{+}$}          & \qw                         & \qw                         & \cdots & \qw            & \ctrl{6}       & \qw                         & \qw                         & \cdots & \qw            & \qw \\
    \lstick[5]{2nd register} & \lstick{$\ket{+}$}          & \qw                         & \qw                         & \cdots & \qw            & \qw            & \ctrl{5}                    & \qw                         & \cdots & \qw            & \qw \\
                             & \lstick{$\ket{+}$}          & \qw                         & \qw                         & \cdots & \qw            & \qw            & \qw                         & \ctrl{4}                    & \cdots & \qw            & \qw \\
                             & \vdots \\
                             & \lstick{$\ket{+}$}          & \qw                         & \qw                         & \cdots & \qw            & \qw            & \qw                         & \qw                         & \cdots & \ctrl{2}       & \qw \\
                             & \lstick{$\ket{+}$}          & \qw                         & \qw                         & \cdots & \qw            & \qw            & \qw                         & \qw                         & \cdots & \qw            & \ctrl{1} \\
                             & \lstick{$\ket{0\cdots 01}$} & \gate{h^{2^{v_1-1}}\bmod p} & \gate{h^{2^{v_1-2}}\bmod p} & \cdots & \gate{h^{2^1}\bmod p} & \gate{h^{2^0}\bmod p} & \gate{g^{2^{v_2-1}}\bmod p} & \gate{g^{2^{v_2-2}}\bmod p} & \cdots & \gate{g^{2^1}\bmod p} & \gate{g^{2^0}\bmod p}
\end{quantikz}
};
\end{tikzpicture}

%\end{adjustbox}
\begin{tikzpicture}
\node[scale=0.6] {
\begin{quantikz}
    & \gate{H} & \ctrl{1}   & \cdots & \ctrl{3}         & \ctrl{4}       & \qw      & \cdots & \qw              & \qw              & \cdots & \qw      & \qw        & \qw      & \meter{$m_0$} \\
    & \qw      & \gate{R_2} & \cdots & \qw              & \qw            & \gate{H} & \cdots & \ctrl{2}         & \ctrl{3}         & \cdots & \qw      & \qw        & \qw      & \meter{$m_1$} \\
    & \vdots \\
    & \qw      & \qw        & \cdots & \gate{R_{v_1-1}} & \qw            & \qw      & \cdots & \gate{R_{v_1-2}} & \qw              & \cdots & \gate{H} & \ctrl{1}   & \qw      & \meter{$m_{v_1-2}$} \\
    & \qw      & \qw        & \cdots & \qw              & \gate{R_{v_1}} & \qw      & \cdots & \qw              & \gate{R_{v_1-1}} & \cdots & \qw      & \gate{R_2} & \gate{H} & \meter{$m_{v_1-1}$} \\
    & \gate{H} & \ctrl{1}   & \cdots & \ctrl{3}         & \ctrl{4}       & \qw      & \cdots & \qw              & \qw              & \cdots & \qw      & \qw        & \qw      & \meter{$m_0'$} \\
    & \qw      & \gate{R_2} & \cdots & \qw              & \qw            & \gate{H} & \cdots & \ctrl{2}         & \ctrl{3}         & \cdots & \qw      & \qw        & \qw      & \meter{$m_1'$} \\
    & \vdots \\
    & \qw      & \qw        & \cdots & \gate{R_{v_2-1}} & \qw            & \qw      & \cdots & \gate{R_{v_2-2}} & \qw              & \cdots & \gate{H} & \ctrl{1}   & \qw      & \meter{$m_{v_2-2}'$} \\
    & \qw      & \qw        & \cdots & \qw              & \gate{R_{v_2}} & \qw      & \cdots & \qw              & \gate{R_{v_2-1}} & \cdots & \qw      & \gate{R_2} & \gate{H} & \meter{$m_{v_2-1}'$} \\
    & \qw      & \qw        & \cdots & \qw              & \qw            & \qw      & \cdots & \qw              & \qw              & \cdots & \qw      & \qw        & \qw      & \qw
\end{quantikz}
};
\end{tikzpicture}
\caption{
%離散対数問題を解くShorアルゴリズムの量子回路。ただし、
Quantum circuit for solving the DLP using Shor's algorithm. Let
$R_j=\begin{bmatrix}
1 & 0 \\
0 & \exp(-2\pi i/2^j)
\end{bmatrix}$.
%上側に載っている前半部分は冪剰余計算、下側に載っている後半部分は逆量子フーリエ変換である。
The first part of the circuit on the top is for modular exponentiation, and the second part of the circuit on the bottom is for the inverse quantum Fourier transform.
}
\label{fig:circuit-standard-dlp}
\end{figure*}
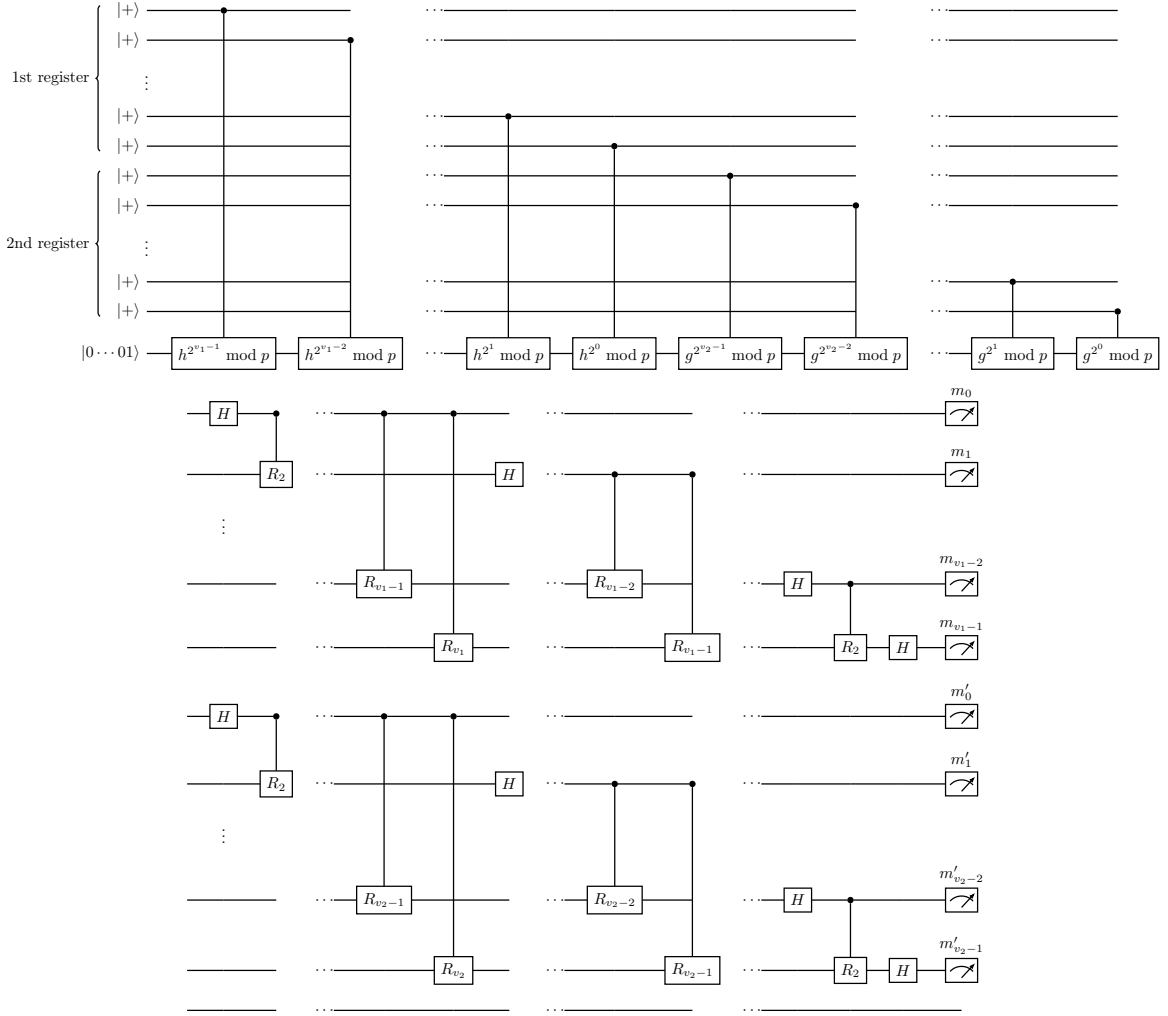

%既存研究\cite{martin2019revisitingshor}にて、uniform superpositionにより半古典逆量子フーリエ変換を利用して第1第2レジスタの量子ビットをrecyclingすることができると述べられている。
In \cite{martin2019revisitingshor}, it is stated that the qubits of the first and second registers can be recycled using the semi-classical quantum Fourier transform with a uniform superposition.
%しかしながらその量子回路は示されていないので、ここで明示しておくと、図\ref{fig:circuit-1qubit-dlp}のようになる。
Since the quantum circuit is not shown, it is explicitly illustrated here as Figure~\ref{fig:circuit-1qubit-dlp}.

\begin{figure*}
\centering
\begin{tikzpicture}
\node[scale=0.6] {
\begin{quantikz}[thin lines]
    \lstick{$\ket{+}$}          & \ctrl{1}                    & \gate{H}   & \meter{$m_0$} & \gate{X^{m_0}} & \gate{H} & \ctrl{1}                    & \gate{S_2} & \gate{H} & \meter{$m_1$} & \gate{X^{m_1}} & \gate{H} & \cdots & \ctrl{1}              & \gate{S_{v_1}} & \gate{H} & \meter{$m_{v_1-1}$} \\
    \lstick{$\ket{0\cdots 01}$} & \gate{h^{2^{v_1-1}}\bmod p} & \qw        & \qw           & \qw            & \qw      & \gate{h^{2^{v_1-2}}\bmod p} & \qw        & \qw      & \qw           & \qw            & \qw      & \cdots & \gate{h^{2^0}\bmod p} & \qw            & \qw      & \qw
\end{quantikz}
};
\end{tikzpicture}

\begin{tikzpicture}
\node[scale=0.6] {
\begin{quantikz}[thin lines]
    \gate{X^{m_{v_1-1}}} & \gate{H} & \ctrl{1}                    & \gate{H}   & \meter{$m_0'$} & \gate{X^{m_0'}} & \gate{H} & \ctrl{1}                    & \gate{S_2'} & \gate{H} & \meter{$m_1'$} & \gate{X^{m_1'}} & \gate{H} & \cdots & \ctrl{1}              & \gate{S_{v_2}'} & \gate{H} & \meter{$m_{v_2-1}'$} \\
    \qw                  & \qw      & \gate{g^{2^{v_2-1}}\bmod p} & \qw        & \qw            & \qw             & \qw      & \gate{g^{2^{v_2-2}}\bmod p} & \qw         & \qw      & \qw            & \qw             & \qw      & \cdots & \gate{g^{2^0}\bmod p} & \qw             & \qw      & \qw
\end{quantikz}
};
\end{tikzpicture}
\caption{
%半古典逆量子フーリエ変換を利用して制御レジスタの量子ビット数を1にしたときの、離散対数問題を解くShorアルゴリズムの量子回路。ただし、
Quantum circuit for solving the DLP using Shor's algorithm when the number of qubits in the control registers is set to 1 by employing the semi-classical inverse quantum Fourier transform. Let
$S_j=\begin{bmatrix}
1 & 0 \\
0 & \exp(-2\pi i\sum_{l=2}^{j} m_{j-l}/2^l)
\end{bmatrix}$ and $S_j'=\begin{bmatrix}
1 & 0 \\
0 & \exp(-2\pi i\sum_{l=2}^{j} m_{j-l}'/2^l)
\end{bmatrix}$.}
\label{fig:circuit-1qubit-dlp}
\end{figure*}
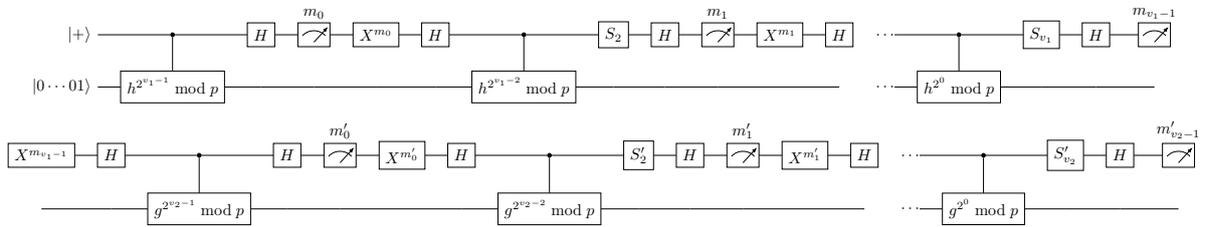

%図\ref{fig:circuit-1qubit-dlp}の方法は制御レジスタの量子ビット数を最小の1にできるが、図\ref{fig:circuit-2qubits-dlp}のように2つにすることで点線で囲われた部分を同時に実行して量子回路の深さを削減することができる。
While the existing method in Figure~\ref{fig:circuit-1qubit-dlp} can minimize the number of qubits in the control register to one, increasing it to two as in Figure~\ref{fig:circuit-2qubits-dlp} allows the parts enclosed by the dashed lines to be executed simultaneously, reducing the time depth of the quantum circuit.
%その最適性についてAppendixにて示している。
The optimality of this is proved in the Appendix.
%\textcolor{red}{（TODO:2つが最適であることの証明をAppendixに置く）}
%さらに$m_0$, ..., $m_0'$, ... を得る論理測定の時間がMod-MULよりも長い場合により高速になる方法もAppendixにて示している。
Additionally, a faster method is shown in the Appendix for cases where the measurement time to obtain $m_0$, ..., $m_0'$, ... is longer than each part of modular exponentiation.
%今回の量子シミュレータ上での実装では、使用したシミュレータ\cite{qulacs,mpiQulacs}の都合上、半古典逆量子フーリエ変換を利用したこれらの方法は利用していない。
In the implementation on the quantum simulator~\cite{qulacs,mpiQulacs} used in this study, these methods utilizing the semi-classical inverse quantum Fourier transform were not employed due to the constraints of the simulator.
%しかし、実際の量子計算機上の実装ではこれらを離散対数問題に限らず素因数分解のShorアルゴリズム全般で活用できる。
However, in actual quantum computer implementations, these methods can be applied not only to the DLP but also to Shor's algorithm for integer factorization in general.

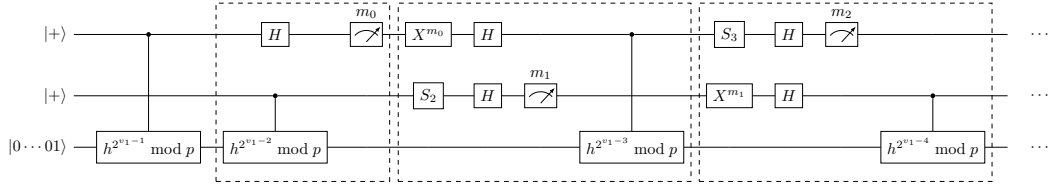
\begin{figure*}
\centering
\begin{tikzpicture}
\node[scale=0.6] {
\begin{quantikz}[thin lines]
    \lstick{$\ket{+}$}          & \ctrl{2}                    & \gate{H}\gategroup[3,steps=2,style={dashed, inner xsep=1pt, inner ysep=8pt}]{} & \meter{$m_0$} & \gate{X^{m_0}}\gategroup[3,steps=4,style={dashed, inner xsep=1pt, inner ysep=8pt}]{} & \gate{H} & \qw           & \ctrl{2}                    & \gate{S_3}\gategroup[3,steps=4,style={dashed, inner xsep=1pt, inner ysep=8pt}]{} & \gate{H} & \meter{$m_2$} & \qw                         & \qw & \cdots \\
    \lstick{$\ket{+}$}          & \qw                         & \ctrl{1}                                                                       & \qw           & \gate{S_2}                                                                           & \gate{H} & \meter{$m_1$} & \qw                         & \gate{X^{m_1}}                                                                   & \gate{H} & \qw           & \ctrl{1}                    & \qw & \cdots \\
    \lstick{$\ket{0\cdots 01}$} & \gate{h^{2^{v_1-1}}\bmod p} & \gate{h^{2^{v_1-2}}\bmod p}                                                    & \qw           & \qw                                                                                  & \qw      & \qw           & \gate{h^{2^{v_1-3}}\bmod p} & \qw                                                                              & \qw      & \qw           & \gate{h^{2^{v_1-4}}\bmod p} & \qw & \cdots
\end{quantikz}
};
\end{tikzpicture}
\caption{
%半古典逆量子フーリエ変換を利用し、かつ量子回路の深さを削減するために制御レジスタの量子ビット数を2にしたときの、離散対数問題を解くShorアルゴリズムの量子回路。点線で囲われた部分は並列に実行可能
Quantum circuit for solving the DLP using Shor's algorithm when the number of qubits in the control registers is set to 2 by employing the semi-classical inverse quantum Fourier transform and reducing the time depth of the quantum circuit.
The parts enclosed by dashed lines can be executed in parallel.
}
\label{fig:circuit-2qubits-dlp}
\end{figure*}

%しかし、今回は実装に用いたqulacs\cite{qulacs}の都合上その方式はとっていない。
%\textcolor{red}{（とはいえ、書くべきでは？その恩恵として、特許手法もここに盛り込めるのが大きい）}

\subsection{The classical computation part}

%量子計算パートで得られた$\widetilde{sl/q}$, $\widetilde{l/q}$から$s$を求めるのがこのパートの目的である．
The purpose of this part is to determine $s$ from the $\widetilde{sl/q}$ and $\widetilde{l/q}$ obtained in the quantum computation part.
%はじめに$\widetilde{l/q}$から$l\in\mathbb{Z}$を求める．
First, we determine $l\in\mathbb{Z}$ from $\widetilde{l/q}$.
%これは$2^{v_2}>q$であれば，$l=\lceil\widetilde{l/q}\cdot q\rfloor$により得られる．
If $2^{v_2}>q$, $l$ can be obtained by $l=\lceil\widetilde{l/q}\cdot q\rfloor$.
%%これは$v_2\geq\lfloor\log q\rfloor$であれば，$l=\lceil\widetilde{l/q}\cdot q\rfloor$により得られる．
%同様に，$\widetilde{sl/q}$から$2^{v_1}>q$であれば$\beta=\lceil\widetilde{sl/q}\cdot q\rfloor$と求まる．
Similarly, if $2^{v_1}>q$, $\beta$ can be determined from $\beta=\lceil\widetilde{sl/q}\cdot q\rfloor$.
%%また，$\widetilde{sl/q}$から同様に，$v_1\geq\lfloor\log q\rfloor$であれば$\beta=\lceil\widetilde{sl/q}\cdot q\rfloor$と求まる．

%ここまでで$l$, $\beta$が求まっているため，残りの$s$, $\alpha$と合わせて式\eqref{eq:tilde-sl-q}より，次の不定方程式が得られる．
Since $l$ and $\beta$ have been determined up to this point, combining them with $s$ and $\alpha$, the following indeterminate equation is obtained from \eqref{eq:tilde-sl-q}:
\begin{align}
    ls-q\alpha=\beta.
\end{align}
%ここで，$q$は素数であるため，両辺の$\mod q$をとると，
Here, since $q$ is a prime number and known, taking both sides modulo $q$ gives:
\begin{align}
    ls&=\beta\pmod q \\
    s&=\frac{\beta}{l}\pmod q.
\end{align}
%より，$s$は求まる．
Thus, $s$ can be determined.

%従来の研究においては連分数展開を用いていた\cite{martin2019revisitingshor}が、この四捨五入を使う方法でそれよりも高速に求めることができる。
In previous research, continued fraction expansion was used~\cite{martin2019revisitingshor}, but this rounding method allows for a slightly faster determination.

\subsection{Space and Time Complexity}

%量子ビット数で表される空間計算量について，古典計算パートでの$l$, $\beta$を求める過程から第1，第2レジスタの量子ビット数はともに
Regarding the space complexity represented by the number of qubits, the number of qubits in the first and second registers, derived from the process of determining $l$ and $\beta$ in the classical computation part, is given by
\begin{align}
    v_1=v_2=\lfloor\log q\rfloor+1\label{eq:1st-2nd-registers}.
\end{align}
%である．
%この結果は素因数分解の場合と大きく異なっている．
This result is significantly different from the case of factorization.
%合成数$N$を素因数分解するShorのアルゴリズムでは，第1レジスタが持つ値の分母である位数$q$が未知である．
In Shor's algorithm for factoring a composite number $N$, the order $q$, which is the denominator of the value held by the first register after the inverse quantum Fourier transform, is unknown.
%ゆえに，その測定結果から真の有理数を精度良く推定するには第1レジスタの量子ビット数を$2\lceil\log N\rceil+1$用意し，そこから連分数展開アルゴリズムを実行する必要がある\cite{qcqi}．
Therefore, to accurately estimate the true rational number from the measurement result, the first register requires $2\lceil\log N\rceil+1$ qubits, and the continued fraction expansion algorithm must be executed from there~\cite{qcqi}.
%したがって，離散対数問題の場合には第1，第2レジスタの必要とする量子ビット数が式\eqref{eq:1st-2nd-registers}であり，かつ推定に連分数展開ではなく$q$を掛けたうえで四捨五入および不定方程式を用いるということが素因数分解と異なっている．
Consequently, in the case of the DLP, the number of qubits required for the first and second registers is given by \eqref{eq:1st-2nd-registers}, and the estimation uses rounding and an indeterminate equation after multiplying by $q$ instead of continued fraction expansion, which is different from factorization.
%なお、第1レジスタに余分な量子ビットを加えて成功確率を99\%以上にする方法もある\cite{martin2019revisitingshor}が、今回は量子ビット数が増えると指数関数的に負荷のかかる古典計算機上で量子シミュレータを動かしていること、および近い将来の実機実験で採用される規模を想定しているので、その考慮はしない。
Note that there is a method to add extra qubits to the first register to achieve a success probability of over 99\%~\cite{martin2019revisitingshor}, but this consideration is omitted here due to the exponential load increase on classical computers running quantum simulators as the number of qubits increases.%, and the assumption of the scale to be adopted in near-future real quantum machine experiments.

%第3レジスタは$p$未満の値まで格納できれば良いので，その量子ビット数は
The third register only needs to store values less than $p$, so the number of qubits required is
\begin{align}
    w=\lceil\log p\rceil.
\end{align}
%となる．
%また，冪剰余計算のための補助ビットを$A_\text{space}$とする．
Additionally, let $A_\text{space}$ be the auxiliary qubits for modular exponentiation.
%%これは冪剰余にてどの加算手法を採用するかで異なる．
%これは冪剰余にてどの加算手法を採用するかで異なり、具体的な形については次節で議論する。
This varies depending on the addition method adopted in modular exponentiation, and the specific form will be discussed in the next section.
%%また，オラクル計算のための補助ビットが$\lceil\log p\rceil+2$量子ビット必要となる\cite{qadd}．
%よって，空間計算量は
Therefore, the space complexity is
\begin{align}
    %v_1+v_2+w+(\lceil\log p\rceil+2)=2\lceil\log p\rceil+2\lfloor\log q\rfloor+4\label{eq:space-complexity}
    v_1+v_2+w+A_\text{space}=\lceil\log p\rceil+2\lfloor\log q\rfloor+2+A_\text{space}.\label{eq:general-space-complexity}
\end{align}
%である．
%もっとも、Quantum Partで示したように、半古典逆量子フーリエ変換を用いる場合は第1・第2レジスタをまとめて1量子ビットもしくは2量子ビットにできるので、
However, as shown in the Quantum Computation Part, if the semi-classical inverse quantum Fourier transform is used, the first and second registers can be combined into 1 or 2 qubits, so
\begin{align}
    1+w+A_\text{space}=\lceil\log p\rceil+1+A_\text{space}\label{eq:general-1qubit-space-complexity}
\end{align}
%もしくは
or
\begin{align}
    2+w+A_\text{space}=\lceil\log p\rceil+2+A_\text{space}.\label{eq:general-2qubit-space-complexity}
\end{align}
%となる。

%次に，ゲート数で表される時間計算量について記述する．
Next, we describe the time complexity represented by the number of gates.
%冪剰余の計算は第1，第2レジスタそれぞれで実行する．
The modular exponentiation is executed with the first and second registers as control. 
%%この冪剰余計算の時間計算量を$A_\text{time}$とすると，
%%これまでに知られている手法については$m,n\in\mathbb{N}$について，$A_\text{time}=O(\lceil\log p\rceil^m \lfloor\log q\rfloor^n)$と表される\cite{kunihiro}．
%素因数分解においては冪剰余計算の時間計算量は、これまでに知られているどの手法でも素因数分解対象の値のビット数の多項式オーダで表される\cite{kunihiro}。
In factorization, the time complexity of modular exponentiation is expressed in polynomial order of the bit length of the value to be factored by any known method~\cite{kunihiro}.
%そこから、離散対数問題においても同様に、その時間計算量を$A_\text{time}$と表すと、$A_\text{time}=O(\lceil\log p\rceil^m \lfloor\log q\rfloor^n)$と書ける。
From this, similarly in the DLP, when the time complexity of modular exponentiation is expressed as $A_\text{time}$, it can be written as $A_\text{time}=O(\lceil\log p\rceil^m \lfloor\log q\rfloor^n)$.
%空間計算量と同様に、各加算手法を採用した際の $m$, $n$ については次節で議論する。
As with space complexity, the values of $m$ and $n$ for each addition method will be discussed in the next section.
%%そこでは\cite{qadd}のcontrolled-$U_a$をそのビット数$\lfloor\log q\rfloor+1$の回数だけ実行することになり，
%そして，現在知られているいずれの加算手法を利用しようとも，冪剰余計算の時間計算量がアルゴリズム全体で最も支配的である\cite{kunihiro,Gidney2021howtofactorbit}ので，全体の時間計算量は
Since the time complexity of modular exponentiation is the most dominant in the entire algorithm regardless of the addition method used~\cite{kunihiro,Gidney2021howtofactorbit}, the overall time complexity is
\begin{align}
    %O(\lceil\log p\rceil^3\lfloor\log q\rfloor)\label{eq:time-complexity}
    A_\text{time}=O(\lceil\log p\rceil^m \lfloor\log q\rfloor^n).\label{eq:general-time-complexity}
\end{align}
%である．

\section{Construction of Quantum Circuit}

%\textcolor{red}{（どこまでR-ADD, Q-ADDを細かく書くか悩んでいる）}
%
%Shorアルゴリズムの量子回路のうち最も計算量の大きいプロセスである冪剰余計算は、Mod-EXP回路により実現される\cite{kunihiro}。
Among the quantum circuits of Shor's algorithm, the modular exponentiation, which is the most computationally intensive process, is realized by the Mod-EXP circuit~\cite{kunihiro}.
%そのうちの加算器の実装方法により、必要な量子ビット数およびゲート数が変化する。
The number of required qubits and gates varies depending on the implementation method of the adder.
%加算器の実装方法としてはQ-ADD, GT-ADD, R-ADDの3種類が存在し、Q-ADDとGT-ADDは必要な量子ビット数が少ない代わりにゲート数が増大し、
%一方、R-ADDは必要な量子ビット数が多い代わりにゲート数が少なく済む\cite{kunihiro}。
There are three types of adder implementations: Q-ADD, GT-ADD, and R-ADD.
Q-ADD and GT-ADD require fewer qubits but result in an increased number of gates, whereas R-ADD requires more qubits but fewer gates~\cite{kunihiro}.
%量子シミュレータ上では量子ビット数が増えるほど指数関数的に必要なメモリが増大して実行が困難になることから、今回は量子シミュレータ上で動かす量子回路の作成にはQ-ADDを利用した。
On a quantum simulator, as the number of qubits increased, the required memory grows exponentially, making execution difficult.
Therefore, for creating quantum circuits to run on a quantum simulator, Q-ADD was used.
%そして、2048ビット規模でのリソース見積もりにはゲート数を少なく抑えて実行時間を短くできるR-ADDを利用した。
For resource estimation at the 2048 bit scale, R-ADD, which can reduce the number of gates and shorten execution time, was used.

%Q-ADDとR-ADDのどちらをベースにした回路であっても、\cite{kunihiro,yamaguchi,yamaguchi2023experimentsresource}と同様に、Mod-EXPはMod-MUL, Mod-PS, Mod-ADD, and ADDから構成される：
Regardless of whether the circuit is based on Q-ADD or R-ADD, as in \cite{kunihiro,yamaguchi,yamaguchi2023experimentsresource}, Mod-EXP consists of Mod-MUL, Mod-PS, Mod-ADD, and ADD:
\begin{itemize}
    \item Mod-EXP$(a)$: $\ket{x}\ket{1}\to \ket{x}\ket{a^x\bmod p}$
    \item Mod-MUL$(d)$: $\ket{y}\to \ket{dy\bmod p}$
    \item Mod-PS$(d)$: $\ket{y}\ket{t}\to \ket{y}\ket{t+dy\bmod p}$
    \item Mod-ADD$(d)$: $\ket{y}\to \ket{y+d\bmod p}$
    \item ADD$(d)$: $\ket{y}\to \ket{y+d}$
\end{itemize}
%%一方、素因数分解ではなく離散対数問題を解くためのShorアルゴリズムの量子回路を構築するので、必要な量子ビット数やゲート数に変化が生じる。
%Constructing a quantum circuit for Shor's algorithm to solve the DLP, rather than factorization, results in changes in the number of required qubits and gates.
%%この節の残りでは、それらが離散対数問題ではどのようなオーダーをとるかを論じる。
%The remainder of this section discusses the order of these requirements for the DLP.
%ここでは素因数分解ではなくDLP向けの量子回路を構築するが、その構成方法自体はほとんど共通している。
Here, we construct quantum circuits for DLP rather than for factorization, but the construction method itself is mostly common.
%とはいえ、素因数分解と異なり変数が$p$と$q$の2つあるので、前節で導入した$A_\textrm{space}$と$A_\textrm{time}$がどのように記述されるかをQ-ADD based circuitとR-ADD based circuitそれぞれについて説明する。
However, unlike factorization, there are two variables, $p$ and $q$, so we will explain how $A_\textrm{space}$ and $A_\textrm{time}$ introduced in the previous section are described for both Q-ADD based circuits and R-ADD based circuits.
\subsection{Q-Add Based Circuits}

%Q-ADDは量子フーリエ変換を加算器に用いる方法である\cite{qadd, thomas2000addition}。
Q-ADD is a method that uses quantum Fourier transform as an adder~\cite{qadd, thomas2000addition}.
%Q-ADDを用いた冪剰余計算回路の構築方法は\cite{yamaguchi,yamaguchi2023experimentsresource}と同一である。
The construction method for modular exponentiation circuits using Q-ADD is the same as in \cite{yamaguchi,yamaguchi2023experimentsresource}.
%すると、Mod-ADDを実現する際に標数 $p$ の補数の範囲の補助ビット数および追加の1ビットが必要になる（詳細は\cite{qadd}のFig.5参照）ので、
Therefore, when implementing Mod-ADD, additional bits for the range of the complement of the modulo $p$ and an extra bit are required (see Figure~5 in \cite{qadd} for details), resulting in
\begin{align}
    A_\text{space}=(\lceil \log p\rceil+1)+1=\lceil \log p\rceil+2.\label{eq:qadd-space-complexity}
\end{align}
%となる。
%ゆえに、全体として必要な量子ビット数は式\eqref{eq:general-space-complexity}より、
%%今回は離散対数問題を扱うので、全体として必要な量子ビット数は
%\begin{align}
%    2\lceil\log p\rceil+2\lfloor\log q\rfloor+4\label{eq:space-complexity}
%\end{align}
%となる。
%一方、時間計算量については\cite{qadd,yamaguchi,yamaguchi2023experimentsresource}と同様に、標数 $p$ の範囲でのMod-MULを一度実行するのに $O(\lceil\log p\rceil^3)$ かかるが、
%その実行回数は $v_1+v_2=2\lceil\log q\rceil+2=O(\lceil\log q\rceil)$ なので、全体の時間計算量は式\eqref{eq:general-time-complexity}より
Regarding time complexity, as in \cite{qadd,yamaguchi,yamaguchi2023experimentsresource}, it takes $O(\lceil\log p\rceil^3)$ to execute Mod-MUL once within the range of the modulo $p$,
and since the number of executions is $v_1+v_2=2\lceil\log q\rceil+2=O(\lceil\log q\rceil)$, the overall time complexity is given by \eqref{eq:general-time-complexity} as
%ゲート数は
\begin{align}
    O(\lceil\log p\rceil^3 \lfloor\log q\rfloor).\label{eq:qadd-time-complexity}
\end{align}
%となる。
%
%
%
\subsection{R-Add Based Circuits}

%R-ADD is a ripple carry adderである\cite{radd}。
R-ADD is a ripple carry adder~\cite{radd}.
%表面符号上でShor algorithmのリソース見積もりをした先行研究でもこの加算手法をベースとしていた\cite{Gidney2021howtofactorbit}。
Previous studies estimating the resources of Shor's algorithm on surface codes also based their addition method on this~\cite{Gidney2021howtofactorbit}.
%As in Q-ADD, R-ADDを用いた冪剰余計算回路の構築方法は\cite{yamaguchi,yamaguchi2023experimentsresource,kunihiro}と同様であり、
%Q-ADDでも必要であった量子ビット数に加えて$\lceil\log p\rceil$個の量子ビット数が必要なため、
As in Q-ADD, the construction method of the modular exponentiation circuit using R-ADD is similar to \cite{yamaguchi,yamaguchi2023experimentsresource,kunihiro}, and in addition to the number of qubits required for Q-ADD, $\lceil\log p\rceil$ qubits are needed, resulting in
\begin{align}
    A_\text{space}=(\lceil \log p\rceil+2)+\lceil\log p\rceil=2\lceil\log p\rceil+2.\label{eq:radd-space-complexity}
\end{align}
%となる。
%%一方、ゲート数に関わる部分では素因数分解と異なる点が存在する。
%一方、ゲート数に関わる部分では素因数分解と異なる点が存在するが、計算量オーダーには影響しないのでAppendixにて記す。
While there are differences in the part related to the number of gates compared to factorization, but they do not affect the order of time complexity, so they are described in the Appendix.
%すると、\cite{kunihiro,yamaguchi,yamaguchi2023experimentsresource}と同様に、標数 $p$ の範囲でのMod-MULを一度実行するのに $O(\lceil\log p\rceil^2)$ かかり、
%それが $v_1+v_2=O(\lceil\log q\rceil)$ 回おこなわれるので、全体の時間計算量は
Thus, as in \cite{kunihiro,yamaguchi,yamaguchi2023experimentsresource}, it takes $O(\lceil\log p\rceil^2)$ to execute Mod-MUL once in the range of modulo $p$, and since this is done $v_1+v_2=O(\lceil\log q\rceil)$ times, the overall time complexity is
%%もっとも、計算量オーダーには影響せず、\cite{kunihiro,yamaguchi,yamaguchi2023experimentsresource}と同様に、標数 $p$ の範囲でのMod-MULを一度実行するのに $O(\lceil\log p\rceil^2)$ かかり、
\begin{align}
    O(\lceil\log p\rceil^2 \lfloor\log q\rfloor).\label{eq:radd-time-complexity}
\end{align}
%となる。

%\begin{enumerate}
%    \item dirty, clean量子ビットの扱いまわり
%    \begin{enumerate}
%        \item factoringと異なり、レジスタが2つあるので、それらを利用したclean量子ビットを使える
%    \end{enumerate}
%\end{enumerate}

\section{Simulating Quantum Circuits}

%\subsubsection{Simulating Shor Algorithm Circuits}

%離散対数問題を解くShorのアルゴリズムを量子シミュレータ上で実行した結果を示す．
The results of executing Shor's algorithm for solving the DLP on a quantum simulator are presented.
%生成した回路の実行を伴う場合は前述したように必要な量子ビット数の少ないQ-ADDを利用し、
When executing the generated circuits, we utilized Q-ADD, which requires fewer qubits as previously mentioned,
and employed mpiQulacs~1.3.1~\cite{mpiQulacs}, a state vector simulator developed by Fujitsu.
We used a PRIMEHPC FX700 with one node for up to 28 qubits, four nodes for 29-32 qubits, and up to 256 nodes for 36 qubits.
%富士通製の状態ベクトルシミュレータであるmpiQulacs 1.3.1~\cite{mpiQulacs}を利用し，PRIMEHPC FX700を28量子ビットまでは1台，29-32量子ビットでは4台，36量子ビットでは256台用いた．
Additionally, after generating the quantum circuits, we applied \lstinline$optimize_light$ from mpiQulacs as an optimization step.
%また，量子回路を生成した後に最適化として，mpiQulacsの\lstinline$optimize_light$を施している．
This optimization not only includes the greedy method from Qulacs~\cite{qulacs} but also reduces inter-process communication by fused-swap gates~\cite{mpiQulacs}.
%これはQulacs~\cite{qulacs}の貪欲法による最適化に加えて，fused-swapゲートを利用したプロセス間通信量の削減を実現している\cite{mpiQulacs}．

%一方、リソース見積もりのための回路生成時には、必要なゲート数が少なく済むR-ADDを利用している。
On the other hand, for circuit generation aimed at resource estimation, we utilized R-ADD, which requires fewer gates.
%そして、最適化には通常のQulacs~\cite{qulacs}の\lstinline$optimize_light$を用いている。
For optimization, we used the standard \lstinline$optimize_light$ from Qulacs~\cite{qulacs}.

\subsection{Running Generated Circuits}

%\textcolor{red}{（パターン3については削除予定。成功確率についてはこの形でいく予定）}

%32量子ビットまでの
We executed quantum circuit simulations using Q-ADD for all 1,860 combinations of prime numbers $p$ and $q$ that satisfy
\begin{align}
    q|(p-1)
\end{align}
%を満たすすべての素数$p$, $q$の組合せ1860通りについて，Q-ADDを利用してその量子回路シミュレーションを実行した．
up to 32 qubits.
%この全組合せの一覧を表\ref{tab:pq-all-combinations}に示す．
The complete list of these combinations is shown in Table~\ref{tab:pq-all-combinations}.
%シミュレーションの結果，すべての組合せにて正しく解$s$は求められた．
The simulation results confirmed that the correct solution $s$ was obtained for all combinations.

%\begin{enumerate}
%    \item 32量子ビットまでの1860通りのすべての $p,q$ の組合せ（表\ref{tab:pq-all-combinations}）についてQ-ADDを利用して網羅的に実験
%    \item すべての組合せで正しく解 $s$ は求められた
%    \item 各組合せについて10000サンプル実験した際の成功確率は概ね40\%から90\%に分布していた。特に $q$ の値が大きいほど60\%以上のもののみになっていった\textcolor{red}{（\cite{martin2019revisitingshor}の成功確率解析と合わせて後で考察）}
%    \item より大きな規模の問題として36量子ビットのパターン1,3について実験し、いずれも正しく解が求まった
%\end{enumerate}

\begin{table*}[tb]\tiny
    \caption{
        %32量子ビットまでの1860通りの可能なすべての$p$と$q$の組合せ．各行は対象の$p$の最小値・最大値の範囲内における，$q$のそれぞれの値の個数を表している．例えば$p\in[300,600]$における$q=5$となる$p$, $q$の組合せの個数は10個である．また，--は32量子ビットまでで表現することのできない$p$, $q$の組合せを表す．
        All 1,860 possible combinations of $p$ and $q$ up to 32 qubits.
        Each row represents the number of values $q$ within the minimum and maximum range of the target $p$.
        For example, the number of combinations of $p$ and $q$ where $q=5$ and $p\in[300,600]$ is 10.
        Additionally, -- indicates combinations of $p$ and $q$ that cannot be represented within 32 qubits.
        %パターン1，3の個数は，その行の$p$の範囲におけるパターン1，3に該当する$p$, $q$の組合せの個数を表す．
        %\#same-prime groupは，その行の$p$の範囲におけるsame-prime groupに該当する$p$, $q$の組合せの個数を表す．
        The \#same-prime group represents the number of combinations of $p$ and $q$ that fall into the same-prime group within the range of $p$ in that row.
    }
    \centering
    %\tabrowsep = 1pt
    %\renewcommand{\arraystretch}{0.2}
    \begin{tabular}{|ll|llllllllllllllllllll|l|l|l|}
\hline
%\multicolumn{2}{|c|}{$p$} & \multicolumn{20}{c|}{$q$の個数} & \multirow{2}{5em}{パターン1の個数} & \multirow{2}{5em}{パターン3の個数} & \multirow{2}{*}{累計} \\
\multicolumn{2}{|c|}{$p$} & \multicolumn{20}{c|}{\#$q$} & \multirow{2}{6em}{\#safe-prime group} & \multirow{2}{*}{total} \\
min & max & 2 & 3 & 5 & 7 & 11 & 13 & 17 & 19 & 23 & 29 & 31 & 37 & 41 & 43 & 47 & 53 & 61 & 83 & 89 & 113 & & \\ \hline
0 & 300 & 62 & 28 & 15 & 9 & 4 & 4 & 3 & 2 & 3 & 2 & 0 & 2 & 1 & 1 & 1 & 1 & 0 & 1 & 1 & 1 & 11 & 141 \\ \hline
300 & 600 & 47 & 22 & 10 & 7 & 5 & 5 & 3 & 3 & 2 & 2 & 2 & 0 & 0 & 1 & 0 & 0 & 1 & -- & -- & -- & 0 & 251 \\ \hline
600 & 900 & 45 & 24 & 11 & 9 & 6 & 2 & 2 & 2 & 2 & 0 & 1 & -- & -- & -- & -- & -- & -- & -- & -- & -- & 0 & 355 \\ \hline
900 & 1200 & 42 & 20 & 12 & 7 & 4 & 4 & 3 & 0 & 2 & 1 & 0 & -- & -- & -- & -- & -- & -- & -- & -- & -- & 0 & 450 \\ \hline
1200 & 1500 & 43 & 21 & 10 & 6 & 4 & 5 & -- & -- & -- & -- & -- & -- & -- & -- & -- & -- & -- & -- & -- & -- & 0 & 539 \\ \hline
1500 & 1800 & 39 & 22 & 7 & 5 & 2 & 1 & -- & -- & -- & -- & -- & -- & -- & -- & -- & -- & -- & -- & -- & -- & 0 & 615 \\ \hline
1800 & 2100 & 39 & 17 & 10 & 5 & 2 & 5 & -- & -- & -- & -- & -- & -- & -- & -- & -- & -- & -- & -- & -- & -- & 0 & 693 \\ \hline
2100 & 2400 & 40 & 21 & 12 & 8 & -- & -- & -- & -- & -- & -- & -- & -- & -- & -- & -- & -- & -- & -- & -- & -- & 0 & 774 \\ \hline
2400 & 2700 & 36 & 16 & 8 & 8 & -- & -- & -- & -- & -- & -- & -- & -- & -- & -- & -- & -- & -- & -- & -- & -- & 0 & 842 \\ \hline
2700 & 3000 & 37 & 16 & 8 & 6 & -- & -- & -- & -- & -- & -- & -- & -- & -- & -- & -- & -- & -- & -- & -- & -- & 0 & 909 \\ \hline
3000 & 3300 & 33 & 18 & 10 & 5 & -- & -- & -- & -- & -- & -- & -- & -- & -- & -- & -- & -- & -- & -- & -- & -- & 0 & 975 \\ \hline
3300 & 3600 & 40 & 22 & 11 & 7 & -- & -- & -- & -- & -- & -- & -- & -- & -- & -- & -- & -- & -- & -- & -- & -- & 0 & 1055 \\ \hline
3600 & 3900 & 36 & 19 & 8 & 6 & -- & -- & -- & -- & -- & -- & -- & -- & -- & -- & -- & -- & -- & -- & -- & -- & 0 & 1124 \\ \hline
3900 & 4200 & 35 & 17 & 6 & 2 & -- & -- & -- & -- & -- & -- & -- & -- & -- & -- & -- & -- & -- & -- & -- & -- & 0 & 1184 \\ \hline
4200 & 4500 & 36 & 15 & -- & -- & -- & -- & -- & -- & -- & -- & -- & -- & -- & -- & -- & -- & -- & -- & -- & -- & 0 & 1235 \\ \hline
4500 & 4800 & 36 & 19 & -- & -- & -- & -- & -- & -- & -- & -- & -- & -- & -- & -- & -- & -- & -- & -- & -- & -- & 0 & 1290 \\ \hline
4800 & 5100 & 35 & 17 & -- & -- & -- & -- & -- & -- & -- & -- & -- & -- & -- & -- & -- & -- & -- & -- & -- & -- & 0 & 1342 \\ \hline
5100 & 5400 & 31 & 13 & -- & -- & -- & -- & -- & -- & -- & -- & -- & -- & -- & -- & -- & -- & -- & -- & -- & -- & 0 & 1386 \\ \hline
5400 & 5700 & 38 & 22 & -- & -- & -- & -- & -- & -- & -- & -- & -- & -- & -- & -- & -- & -- & -- & -- & -- & -- & 0 & 1446 \\ \hline
5700 & 6000 & 33 & 15 & -- & -- & -- & -- & -- & -- & -- & -- & -- & -- & -- & -- & -- & -- & -- & -- & -- & -- & 0 & 1494 \\ \hline
6000 & 6300 & 36 & 18 & -- & -- & -- & -- & -- & -- & -- & -- & -- & -- & -- & -- & -- & -- & -- & -- & -- & -- & 0 & 1548 \\ \hline
6300 & 6600 & 34 & 18 & -- & -- & -- & -- & -- & -- & -- & -- & -- & -- & -- & -- & -- & -- & -- & -- & -- & -- & 0 & 1600 \\ \hline
6600 & 6900 & 34 & 18 & -- & -- & -- & -- & -- & -- & -- & -- & -- & -- & -- & -- & -- & -- & -- & -- & -- & -- & 0 & 1652 \\ \hline
6900 & 7200 & 32 & 13 & -- & -- & -- & -- & -- & -- & -- & -- & -- & -- & -- & -- & -- & -- & -- & -- & -- & -- & 0 & 1697 \\ \hline
7200 & 7500 & 31 & 17 & -- & -- & -- & -- & -- & -- & -- & -- & -- & -- & -- & -- & -- & -- & -- & -- & -- & -- & 0 & 1745 \\ \hline
7500 & 7800 & 37 & 19 & -- & -- & -- & -- & -- & -- & -- & -- & -- & -- & -- & -- & -- & -- & -- & -- & -- & -- & 0 & 1801 \\ \hline
7800 & 8100 & 31 & 13 & -- & -- & -- & -- & -- & -- & -- & -- & -- & -- & -- & -- & -- & -- & -- & -- & -- & -- & 0 & 1845 \\ \hline
8100 & 8400 & 10 & 5 & -- & -- & -- & -- & -- & -- & -- & -- & -- & -- & -- & -- & -- & -- & -- & -- & -- & -- & 0 & 1860 \\ \hline
\end{tabular}
    \label{tab:pq-all-combinations}
\end{table*}

%また、より大きな数での実験として、
%%パターン1，3について36量子ビットで表現できる規模の離散対数問題を量子シミュレータで解いた．
%Safe-prime groupについて36量子ビットで表現できる規模である
%$p=503$, $q=251$つまり$p$が9ビット，$q$が8ビットの離散対数問題
%を量子シミュレータで解いた。
Additionally, as an experiment with large numbers, we solved the DLP of a safe-prime group with $(p,q)=(263,131)$, $(347,173)$, $(359,179)$, $(383,191)$, $(467,233)$, $(479,239)$, $(503,251)$, which can be represented with 36 qubits, using a quantum simulator.
Here, $p$ is 9 bits and $q$ is 8 bits.
%%パターン1については$p=503$, $q=251$つまり$p$が9ビット，$q$が8ビットの離散対数問題，
%%パターン3については$p=4093$, $q=31$つまり$p$が12ビット，$q$が5ビットの離散対数問題
%%を用いた．
%%なお，パターン3の値が式\eqref{eq:qadd-pattern3}に従っていないように見えるが，これはまだ$p,q\gg 1$といえず，$\lceil\log p\rceil\not\approx \log p$および$\lfloor\log q\rfloor\not\approx \log q$となるためである．
%一方，Schnorr groupについては，数体篩法とBaby-step Giant-stepの計算量が一致するように$p$と$q$を決めると36量子ビットの規模では$p<q$となり破綻するため，今回の実験では用いていない．
On the other hand, for the Schnorr group, if $p$ and $q$ are chosen such that the time complexity of the NFS and the Baby-step Giant-step methods match, $p$ becomes less than $q$ at the scale of 36 qubits, causing a breakdown.
Therefore, the Schnorr group was not used in this experiment.
%%結果として、ともに正しく解 $s$ を得ることができ、最適化後の量子回路をシミュレートして10000サンプル得るまでの時間は、パターン1では47分、パターン3では30分であった。
%結果として、正しく解 $s$ を得ることができ、最適化後の量子回路をシミュレートして10000サンプル得るまでの時間は47分であった。
%結果として、正しく解 $s$ を得ることができ、最適化後の量子回路をシミュレートして10000サンプル得るまでの時間は30分～103分であった。
As a result, the correct solution $s$ was obtained, and the time to simulate the optimized quantum circuit until 10,000 samples were obtained ranged from 30 to 103 minutes.
%As a result, the correct solution $s$ was obtained, and the time to simulate the optimized quantum circuit until 10,000 samples were obtained was 47 minutes.
%%その結果を比較したものを表\ref{tab:compare-gates}に示す．
%%双方ともに用いた量子ビット数は等しいが，確かにパターン3の方が解くのに必要な量子ゲート数が多い．
%%特に，その傾向は最適化前後で変化しないことも，また，深さにおいて成り立つことも確認できる．

%各 $p$, $q$ の組合せについて10000サンプル実験し、測定結果からClassical Computation Partの手続きに沿って正しく解 $s$ を得られた成功確率をそれぞれ計算した。
For each combination of $p$ and $q$, 10,000 samples were experimented, and the success probability of correctly obtaining the solution $s$ from the measurement results was calculated according to the procedures of the Classical Computation Part.
%その1つ1つの結果を図\ref{fig:success-prob-heatmap}に表し、また、$q$を横軸に成功確率の変化を箱ひげ図として図\ref{fig:success-prob-boxplot}に表した。
Each result is shown in Figure~\ref{fig:success-prob-heatmap}, and the change in success probability with $q$ as the horizontal axis is shown as a box plot in Figure~\ref{fig:success-prob-boxplot}.
%%成功確率は $p$, $q$ の値によって様々な値をとるが、$q$ が大きくなるほど成功確率は 1 に近づく傾向にある。
%%\cite{martin2019revisitingshor}にて、位数 $q$ が大きくなるほど成功確率は上昇すると理論的に示されており、この結果と合致している。

%はじめに、成功確率の最小値と最大値に注目する。
First, let's focus on the minimum and maximum values of the success probability.
%成功確率の最小値は$q=3$で確認された0.4181であり、最大値は$q=13$で確認された0.9133である。
The minimum success probability was 0.4181, confirmed at $q=3$, and the maximum was 0.9133, confirmed at $q=13$.
%%前者を構成する$p$は161通り存在し、全体の1027通りの$p$のうち$q=3$をペアにもつ$p$は505通り存在するので、そのうち$161/505\approx 31.9\%$の成功確率が最小値の0.4181になっている。
%%$q=3$をペアにもつ$p$は全部で505通り存在し、そのうち最小の成功確率となった$p$は161通り、つまり$161/505\approx 31.9\%$である。
%$q=3$をペアにもつ$p$は全部で505通り存在し、そのうち最小の成功確率となった$p$は161通りなので、$161/505\approx 31.9\%$が該当することになる。
There are 505 combinations of $p$ paired with $q=3$, and 161 of them had the minimum success probability, so $161/505\approx 31.9\%$ correspond to this.
%%一方、後者を構成する$p$は1通りの$p=1483$のみであった。
%一方、最大値を構成する$p$は1通りの$p=1483$のみであった。
On the other hand, the maximum value was constituted by only $p=1483$.
%$q=13$には成功確率が90\%を超えるペアはもう1つ存在し、それは$p=1327$での成功確率0.9124である。
There is another pair with $q=13$ that has a success probability exceeding 90\%, which is $p=1327$ with a success probability of 0.9124.
%このような最小値と最大値は理論\cite{martin2019revisitingshor}でもこれまでのシミュレーション\cite{mandl2022implementationsshordlp}でも見つけられていない。
Such minimum and maximum values have not been found in previous theory~\cite{martin2019revisitingshor} or simulations~\cite{mandl2022implementationsshordlp}.
%これは、今回$p$だけでなく$q$についても変化させたすべての組合せについて実験したおかげである。
This is thanks to the experiment conducted with all combinations of not only $p$ but also $q$.
%なお、最大値は図\ref{fig:success-prob-boxplot}における外れ値として現れているが、$q\geq 19$では外れ値が出現していない。
The maximum value appears as an outlier in Figure~\ref{fig:success-prob-boxplot}, but no outliers appear for $q\geq 19$.
%これはそのような$q$における外れ値が存在しないことを必ずしも表しているのではなく、シミュレーションできる量子ビット数の制約により$p$を大きくできず、サンプル数が少なくなっているためであるかもしれない。
This does not necessarily indicate the absence of outliers for such $q$, but it may be due to the constraint on the number of qubits that can be simulated, limiting the size of $p$ and reducing the number of samples.

%次に、成功確率の平均値の振る舞いに注目する。
Next, let's focus on the behavior of the average success probability.
%図\ref{fig:success-prob-boxplot}より、図\ref{fig:success-prob-boxplot_p}での$p$と異なり、$q$が大きくなるにつれて成功確率は周期的に変化し、その幅はおよそ60\%から80\%の間である。
%その様子として、$q$が2のべき乗の直前 ($q=13,31,61$) で極小値を迎え、その直後に極大値となっている。
%$q$によって成功確率が周期的に変化し、そして$q$が2のべき乗-1で極小値の約60\%、2のべき乗+1で極大値の約82\%となることがEkeråのヒューリスティックな理論解析から得られているので、今回の結果は確かにそれを検証できている。
%そのうえで、極大値から緩やかに減少していき、極小値に近づくにつれて急峻に下がるという$q$について左右非対称な波形の形を今回は描くことができた。
%From Figure \ref{fig: success-prob-boxplot}, unlike $p $in Figure \ref{fig: success-prob-boxplot_p}, the probability of success varies periodically as $q $increases, with a range of about 60\% to 80\%.
%The minimum value of $q $is reached immediately before the power of 2 ($q=13,31,61$) and the maximum value is reached immediately after.
As shown in Figure \ref{fig:success-prob-boxplot}, unlike $p$ in Figure \ref{fig:success-prob-boxplot_p}, the success probability changes periodically as $q$ increases, with a range of approximately 60\% to 80\%.
Specifically, the success probability reaches a local minimum value just before $q$ is a power of 2 ($q=13, 31, 61$) and then immediately reaches a local maximum value.
%Ekerå's heuristic analysis shows that the probability of success varies periodically with $q$, and that $q$ is about 60\% of the minimum value at a power of 2 minus 1, and about 82\% of the maximum value at a power of 2 plus 1.
Ekerå's heuristic theoretical analysis~\cite{martin2019revisitingshor} shows that the success probability changes periodically with $q$, and that the success probability is approximately 60\% as a local minimum value at values of $q$ that are one less than powers of 2, and approximately 82\% as a local maximum value at values of $q$ that are one more than powers of 2.
This experimental results confirm that heuristic theoretical analysis.
%On top of that, we were able to draw a shape of asymmetrical waveform of $q $that gradually decreases from the maximum value and steeply decreases as it approaches the minimum value.
Furthermore, we were able to depict a waveform with an asymmetric shape for $q$, where the value decreases gradually from the local maximum value and then drops sharply as it approaches the local minimum value.
Note that Theorem~3 in \cite{martin2019revisitingshor} shows that when certain conditions listed there including classical post-processing are met, the success probability asymptotically approaches 1 as $q$ approaches infinity.
%Although we did not meet those conditions in the most naive conditions this time, within the initially feasible range of q, no significant change in the probability range was observed, except for a slight upward trend in the maximum value.
This time, since these conditions are not met and it is under the most naive conditions, no significant change in the probability range was observed, except for a slight upward trend in the maximum value within the initial experimental range of $q$.

\begin{figure*}[h]
    \includegraphics[scale=0.42]{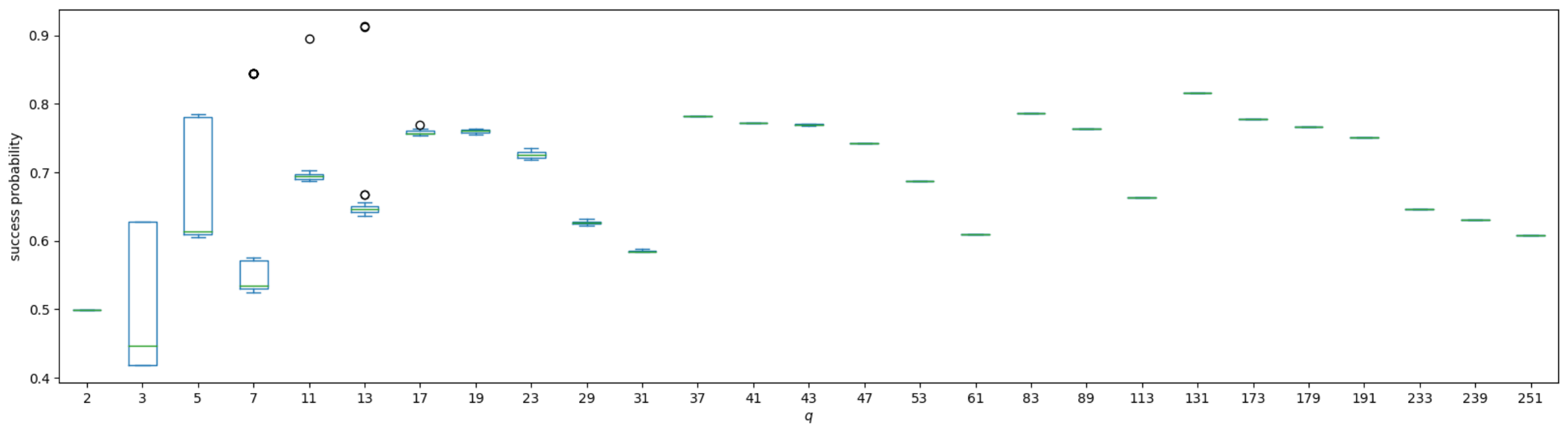}
    \caption{
        %$q$の大きさによる成功確率の変化を表す箱ひげ図
        Box plot showing the change in success probability with the size of $q$.
    }
    \label{fig:success-prob-boxplot}
\end{figure*}

\begin{figure*}[h]
    \includegraphics[scale=0.42]{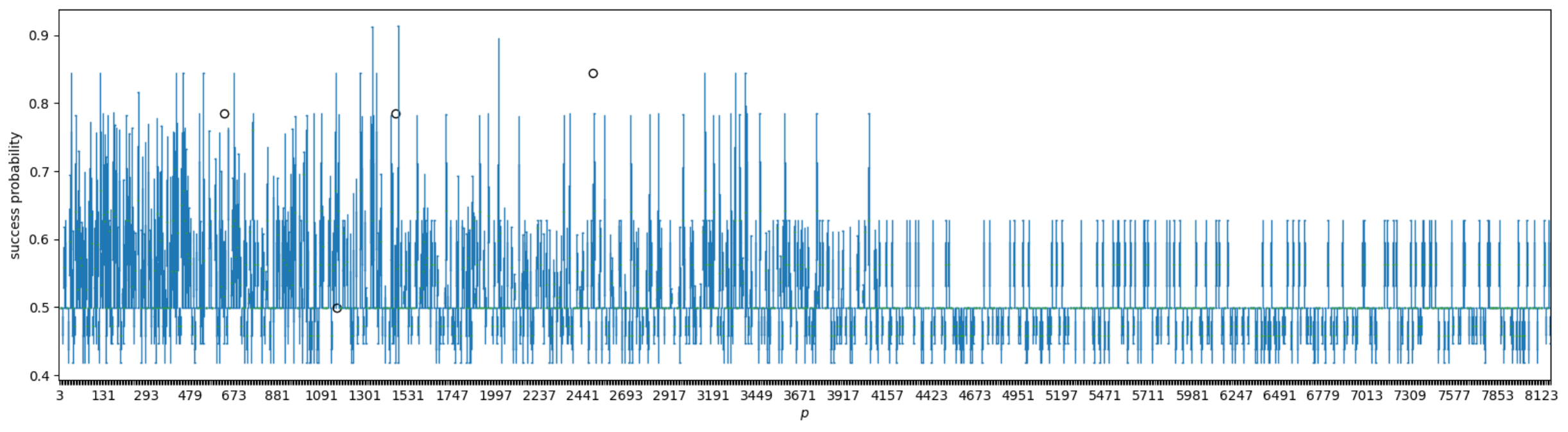}
    \caption{
        %$p$の大きさによる成功確率の変化を表す箱ひげ図
        Box plot showing the change in success probability with the size of $p$.
    }
    \label{fig:success-prob-boxplot_p}
\end{figure*}

\subsection{Resource Estimation by Generated Circuits}

%\textcolor{red}{（Safe-prime groupのpのビット数でSchnorr groupを表すとどこまで減るかを定量的に示すことをメインに変更予定）}

%ゲート数を少なくできるR-ADDによる量子回路を用いてリソース見積もりをおこなう。
We estimate the resources using quantum circuits with R-ADD, which can reduce the number of gates.
%\cite{yamaguchi,yamaguchi2023experimentsresource}では素因数分解する値のビット数ごとに10通りの合成数をランダムに選び、それを解く量子回路を生成してリソース値を得ていた。
In \cite{yamaguchi,yamaguchi2023experimentsresource}, ten composite numbers were randomly selected for each bit length of the values to be factored, and quantum circuits were generated to solve them to obtain resource values.
%そこでの合成数のビット数は8から24までで、結果として実験に用いられた合成数の総数は170個であった。
The bit lengths of the composite numbers ranged from 8 to 24, resulting in a total of 170 composite numbers used in the experiments.
%一方、離散対数問題では標数 $p$ だけでなく位数 $q$ にも自由度がある。
On the other hand, in the DLP, there are degrees of freedom not only in the modulo $p$ but also in the order $q$.
%そこで与えられた量子ビット数で表現できるすべての $p$, $q$ のビット数それぞれについて、さらに5通りずつの具体的な $p$, $q$ の値をランダムに選び、そこから量子回路を生成して実験した。
Therefore, for each bit length of $p$ and $q$ that can be represented by the given number of qubits, five specific values of $p$ and $q$ were randomly selected, and quantum circuits were generated and experimented with.
%例えば、量子ビット数が40のとき、表現できる$p$, $q$のビット数の組合せは式\eqref{eq:radd-space-complexity}、\eqref{eq:general-space-complexity}より、$(\lceil\log p\rceil, \lfloor\log q\rfloor)=(8,6),(10,3)$である。
%そのうち、前者は$(p,q)=(167, 83), (179, 89), (227, 113)$のみが存在するので、その3通りすべてを用いる。
%後者は5種類以上の組合せが存在するので、そこからランダムに$(p,q)=(937, 13), (859, 13), (859, 11), (911, 13), (661, 11)$を選んだ。
%例えば、量子ビット数が41のとき、表現できる$p$, $q$のビット数の組合せは式\eqref{eq:radd-space-complexity}、\eqref{eq:general-space-complexity}より、$(\lceil\log p\rceil, \lfloor\log q\rfloor)=(9,5),(11,2)$である。
For example, when the number of qubits is 41, the combinations of bit lengths of $p$ and $q$ that can be represented are $(\lceil\log p\rceil, \lfloor\log q\rfloor)=(9,5),(11,2)$ according to \eqref{eq:radd-space-complexity} and \eqref{eq:general-space-complexity}.
%そのうち、前者は$(p,q)=(283, 47), (367, 61), (431, 43)$のみが存在するので、その3通りすべてを用いる。
Among them, only $(p,q)=(283, 47), (367, 61), (431, 43)$ exist for the former, so all three combinations were used.
%後者は5種類以上の組合せが存在するので、そこからランダムに$(p,q)=(1871, 5), (1901, 5), (1321, 5), (1597, 7), (1231, 5)$を選んだ。
For the latter, more than five combinations exist, so $(p,q)=(1871, 5), (1901, 5), (1321, 5), (1597, 7), (1231, 5)$ were randomly selected.
%量子ビット数としては12から90までで、結果として実験に用いられた $p$, $q$ の組合せの総数は954個であった。
The number of qubits ranged from 12 to 90, resulting in a total of \numberofonlyoptimizeradd{} combinations of $p$ and $q$ used in the experiments.

%まず，量子ビット数は式\eqref{eq:general-space-complexity}に式\eqref{eq:radd-space-complexity}を代入したものを連続化した
First, the number of qubits matches the surface depicted by the continuous form of the equation obtained by substituting \eqref{eq:radd-space-complexity} into \eqref{eq:general-space-complexity}
\begin{align}
    3\log p+2\log q+4\label{eq:space-complexity-continuous}
\end{align}
%で描かれた面と一致している（図\ref{fig:qubits}）．
(see Figure~\ref{fig:qubits}).
%一方，最適化前後の深さ・最適化前後のnon-Cliffordゲート数は式\eqref{eq:radd-time-complexity}をベースとした
On the other hand, the depth before and after optimization and the number of non-Clifford gates before and after optimization were fitted using the equation based on \eqref{eq:radd-time-complexity}
\begin{align}
    &a(\log p)^2\log q + b(\log p)^2 + c\log p\log q + d\log p \notag\\
    &+ e\log q + f\label{eq:fitting}
\end{align}
%を用い，$a,b,c,d,e,f\in\mathbb{R}$についてSciPy 1.10.1~\cite{scipy}の非線形最小二乗法(\lstinline$curve_fit$)により最適化する方法でフィッティングした．
and optimized for $a,b,c,d,e,f\in\mathbb{R}$ using the nonlinear least squares method (\lstinline$curve_fit$) of SciPy 1.10.1~\cite{scipy}.
%そのフィッティングの精度を評価するため10分割交差検証をおこない、それぞれのfoldにおけるテストデータの平均値と最小二乗誤差について、すべてのfoldにわたって平均をとったものを表\ref{tab:evaluation-fitting-accuracy}に示す。
To evaluate the accuracy of this fitting, 10-fold cross-validation was performed, and the average value and mean squared error of the test data for each fold were averaged across all folds, as shown in Table~\ref{tab:evaluation-fitting-accuracy}.
%このようにいずれも平均値の10分の1程度の誤差に収まっている。
As shown, the errors are within about one-tenth of the average value.
%ここから、ゲート数・回路深さ・non-Cliffordゲート数について最適化の有無にかかわらずいずれについても、式\eqref{eq:fitting}で表すことができ、その違いは係数のみにあるといえる。
From here, regardless of whether optimization is applied or not, the number of gates, circuit depth, and the number of non-Clifford gates can all be expressed by \eqref{eq:fitting}, with the differences being only in the coefficients.

\begin{table}[tb]
    \caption{
        %式\eqref{eq:fitting}によるフィッティングを10分割交差検証のそれぞれの訓練データについておこない、そのテストデータの平均値とテストデータからの二乗平均平方根誤差について、すべてのfoldにわたって平均をとった結果
        The table shows the average of the test data and the root mean squared error (RMSE) from the test data averaged over all folds, obtained by performing fitting using \eqref{eq:fitting} on each training dataset in the 10-fold cross-validation.
    }
    \centering
    \begin{tabular}{|l|l|l|}
        \hline
        & Average & RMSE \\ \hline
        \begin{tabular}{l}
            \#gates \\before optimization 
        \end{tabular} & $7.66\times 10^{6}$ & $4.22\times 10^{5}$ \\ \hline
        \begin{tabular}{l}
            \#gates \\after optimization 
        \end{tabular} & $2.49\times 10^{6}$ & $1.36\times 10^{5}$ \\ \hline
        \begin{tabular}{l}
            depth \\before optimization 
        \end{tabular} & $4.20\times 10^{6}$ & $2.25\times 10^{5}$ \\ \hline
        \begin{tabular}{l}
            depth \\after optimization 
        \end{tabular} & $1.89\times 10^{6}$ & $1.00\times 10^{5}$ \\ \hline
        \begin{tabular}{l}
            \#non-Clifford gates \\before optimization 
        \end{tabular} & $3.56\times 10^{6}$ & $1.95\times 10^{5}$ \\ \hline
        \begin{tabular}{l}
            \#non-Clifford gates \\after optimization
        \end{tabular} & $2.38\times 10^{6}$ & $1.31\times 10^{5}$ \\ \hline
    \end{tabular}
    \label{tab:evaluation-fitting-accuracy}
\end{table}

%訓練データとして831個すべての組合せを用いたときのフィッティングの結果を図\ref{fig:depth}, \ref{fig:nonclifford}に示す。
The fitting results using all \numberofonlyoptimizeradd{} combinations as training data are shown in Figures~\ref{fig:gates}, \ref{fig:depth}, and \ref{fig:nonclifford}.
%このときの各パラメータは表のようになった。
The parameters are as shown in the Table~\ref{tab:fitting-parameters}.

%\begin{enumerate}
%    \item \cite{yamaguchi,yamaguchi2023experimentsresource}では素因数分解する値のビット数ごとに10通りのサンプルをランダムに生成して実験していた。
%    \item 離散対数問題では $p$ だけでなく $q$ にも自由度がある。そこで与えられた量子ビット数で表現できるすべての $p$, $q$ のビット数それぞれについて、さらに5通りずつのサンプルをランダムに生成して実験した。
%    \item 具体的には75量子ビットまでの各組合せ680通り\textcolor{red}{（今も実験は続いているのでもう少し増える）}について実験した
%    \item SciPy 1.10.1~\cite{scipy}の非線形最小二乗法により最適化する方法でフィッティング（図\ref{fig:depth},\ref{fig:nonclifford}）
%    \item フィッティングで得られたパラメータを利用して、$p=2048$ビットでのリソース見積もりをした（表\ref{tab:r-add-2048}）
%\end{enumerate}

\begin{figure}[tb]
    \includegraphics[scale=0.65]{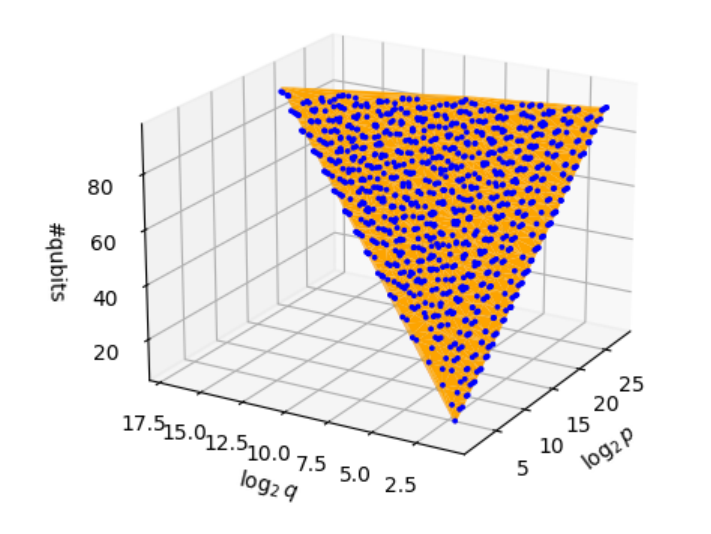}
    \caption{
        %量子ビット数．
        Number of qubits.
        %青い点がそれぞれの実験で得られた結果を表し，オレンジの面が式\eqref{eq:space-complexity-continuous}を表す．
        The blue dots represent the results obtained from each experiment, and the orange surface represents \eqref{eq:space-complexity-continuous}.
    }
    \label{fig:qubits}
\end{figure}

\begin{figure}[tb]
    \includegraphics[scale=0.65]{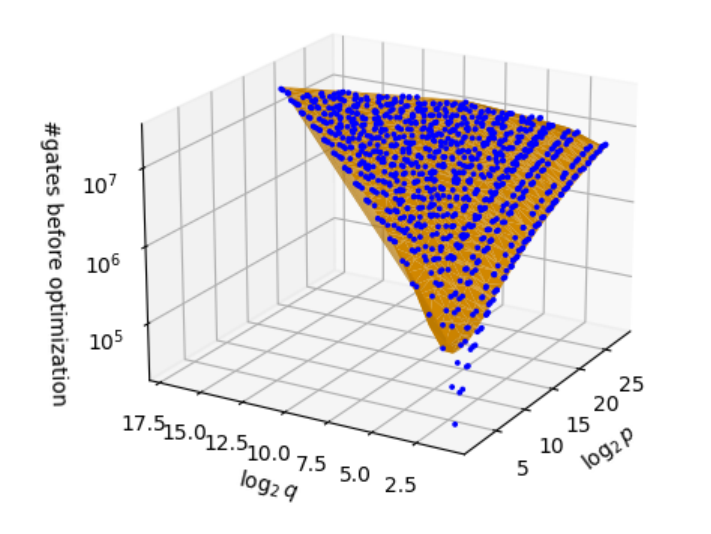}
    \includegraphics[scale=0.65]{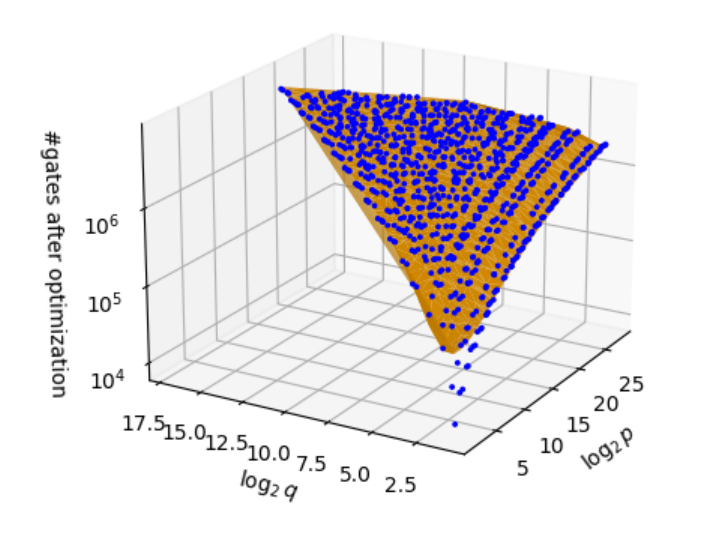}
    \caption{
        %最適化前後のゲート数
        Number of gates before and after optimization.
        The blue dots represent the results obtained from each experiment, and the orange surface represents \eqref{eq:fitting} using the parameters shown in Table~\ref{tab:fitting-parameters}.
    }
    \label{fig:gates}
\end{figure}

\begin{figure}[tb]
    \includegraphics[scale=0.65]{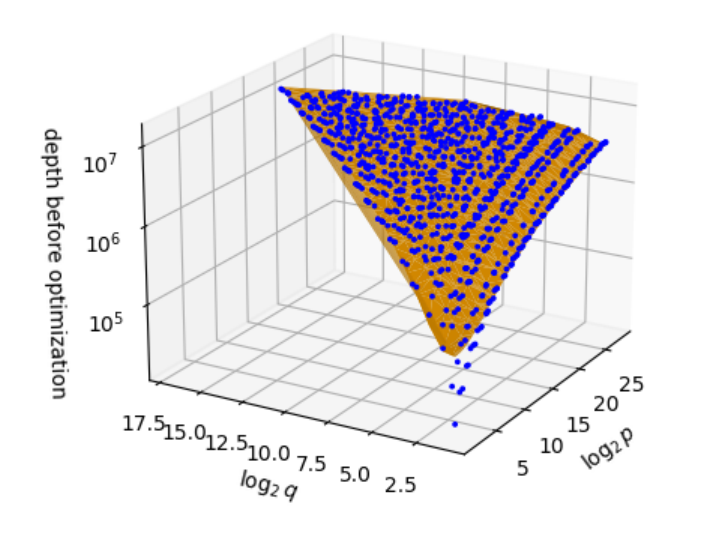}
    \includegraphics[scale=0.65]{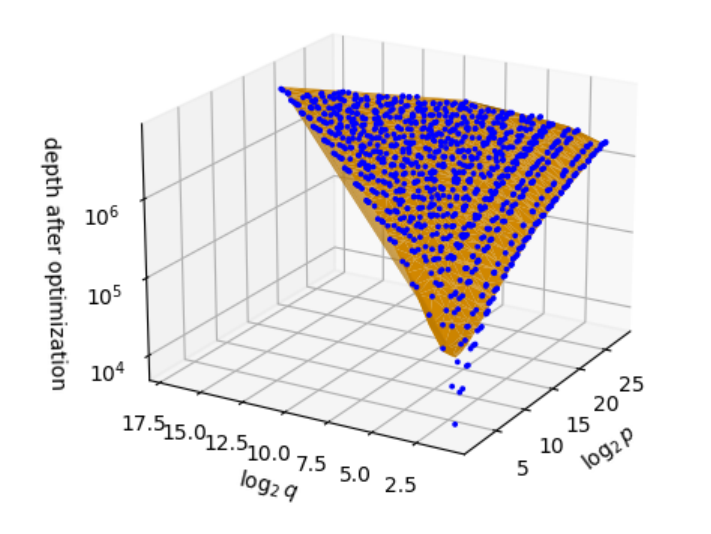}
    \caption{
        %最適化前後の深さ
        Depth before and after optimization.
    }
    \label{fig:depth}
\end{figure}

\begin{figure}
    \includegraphics[scale=0.65]{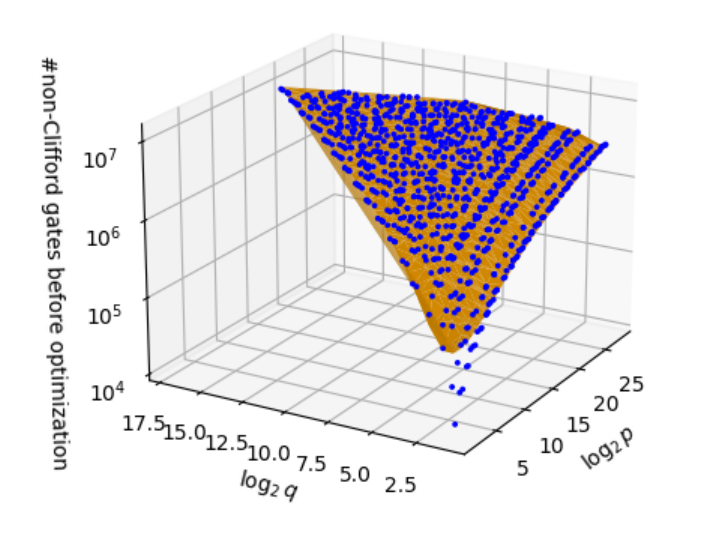}
    \includegraphics[scale=0.65]{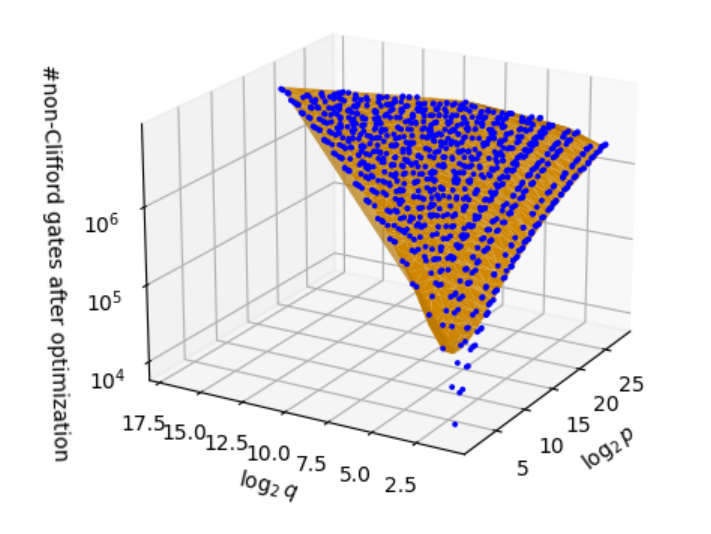}
    \caption{
        %最適化前後のnon-Cliffordゲート数
        Number of non-Clifford gates before and after optimization.
    }
    \label{fig:nonclifford}
\end{figure}

\begin{table*}
    \caption{
        %フィッティングにより得られた式\eqref{eq:fitting}の各パラメータ
        Parameters of \eqref{eq:fitting} obtained from the fitting.
    }
    \centering
    \begin{tabular}{|l|l|l|l|l|l|l|}
        \hline
        & $a$ & $b$ & $c$ & $d$ & $e$ & $f$ \\ \hline
        \#gates before optimization & $2843.058$ & $5260.803$ & $28708.363$ & $-62171.358$ & $-192948.809$ & $459828.645$ \\ \hline
        \#gates after optimization & $923.891$ & $1720.901$ & $9263.616$ & $-19976.457$ & $-62110.960$ & $148050.262$ \\ \hline
        depth before optimization & $1534.728$ & $3178.924$ & $14408.472$ & $-30183.329$ & $-97531.979$ & $228340.308$ \\ \hline
        depth after optimization & $694.800$ & $1421.309$ & $6349.034$ & $-13164.850$ & $-42942.880$ & $100263.303$ \\ \hline
        \#non-Clifford gates before optimization & $1320.030$ & $2444.355$ & $13297.437$ & $-28879.570$ & $-89484.421$ & $213398.262$ \\ \hline
        \#non-Clifford gates after optimization & $882.216$ & $1647.665$ & $8882.046$ & $-19108.039$ & $-59544.637$ & $141771.038$ \\ \hline
    \end{tabular}
    \label{tab:fitting-parameters}
\end{table*}

%この結果を踏まえて，$p$が2048ビットのときのリソースを推定した．
Based on these results, we estimated the resources when $p$ is 2048 bits.
%$p$が2048ビットのとき，Safe-prime groupとSchnorr groupでは$q$のビット数が異なり，Safe-prime groupでは$p=2q+1$を満たすために$q$はおよそ2047ビットとなる．
When $p$ is 2048 bits, the bit length of $q$ differs between safe-prime groups and Schnorr groups.
In safe-prime groups, $q$ is approximately 2047 bits to satisfy $p=2q+1$.
%一方，Schnorr groupでは，$p$により時間計算量が決まる数体篩法と，$q$により時間計算量が決まるBaby-step Giant-step法の2つの時間計算量を一致させるため，$q$は256ビットとなる\cite{dss-256}．
On the other hand, in Schnorr groups, $q$ is 256 bits~\cite{dss-256} to match the time complexity of the NFS, which is determined by $p$, and the Baby-step Giant-step method, which is determined by $q$.
%表\ref{tab:r-add-2048}に示す結果より，Safe-prime groupとSchnorr groupでは必要なゲート数や深さが異なっている．
As shown in Table~\ref{tab:r-add-2048}, the required number of gates and depth differ between safe-prime groups and Schnorr groups.
%古典計算においてはSchnorr groupの$p$, $q$の決め方から明らかなように，$q$が256ビット以上である両者の離散対数問題を解くのに必要な時間計算量は一致する．
In classical computation, as is evident from the way $p$ and $q$ are determined in Schnorr groups, the time complexity required to solve the DLP for both, where $q$ is 256 bits or more, is the same.
%しかし，量子計算ではそうではなく，なぜなら式\eqref{eq:general-time-complexity}のように，全体の時間計算量が$p$と$q$に分離せず両方に依存しているためである．
However, this is not the case in quantum computation because the overall time complexity depends on both $p$ and $q$, as shown in \eqref{eq:general-time-complexity}.
%ここで、量子計算機にとってSchnorr型はどれほど簡単なのかはSame~\#qubits型により理解できる。
%$p$が2048ビットのときのSchnorr型と同じ量子ビット数で表されるSame~\#qubits型の必要とするリソース推定結果も図\ref{tab:r-add-2048}に示している。
%ここから、量子計算機にとっては$p=2048$ビットのSchnorr型は、$p=1479$ビットのSame~\#qubits型未満の難しさであると理解できる。
%%最後に，パターン3では式\eqref{eq:qadd-pattern3}より，$q$は$2048/3\approx 683$ビットとなる．
%ゆえに，DSA署名などで使われているSchnorr型\cite{dss-256}は量子計算機によってSafe-prime型よりも簡単に解かれてしまう．

\begin{table}[tb]
    \caption{
        %R-ADD~\cite{radd}を冪剰余計算に利用した際の，$p$が2048ビットのときのSafe-prime型とSchnorr型におけるリソース推定結果
        Resource estimation results for safe-prime and Schnorr groups when $p$ is 2048 bits, using R-ADD\cite{radd} for modular exponentiation.
        %%、および、Schnorr型と同数の量子ビットで表されるSame~\#qubits型でのリソース推定結果
    }
    \centering
    %\begin{tabular}{|c|l|l|l|}
    \begin{tabular}{|l|l|l|}
        \hline
        ($|p|,|q|$) &
        \begin{tabular}{c}
            Safe-prime\\(2048, 2047)
        \end{tabular}
        &
        \begin{tabular}{c}
            Schnorr\\(2048, 256)
        \end{tabular}
        %&
        %\begin{tabular}{c}
        %    Same~\#qubits\\(1479, 1109)
        %\end{tabular}
        \\ \hline
        \begin{tabular}{l}
            \#qubits
        \end{tabular} & 10241 & 6659 \\ \hline
        \begin{tabular}{l}
            \#gates \\before optimization
        \end{tabular} & $2.46\times 10^{13}$ & $3.09\times 10^{12}$ \\ \hline
        \begin{tabular}{l}
            \#gates \\after optimization
        \end{tabular} & $7.98\times 10^{12}$ & $1.00\times 10^{12}$ \\ \hline
        \begin{tabular}{l}
            depth \\before optimization
        \end{tabular} & $1.33\times 10^{13}$ & $1.67\times 10^{12}$ \\ \hline
        \begin{tabular}{l}
            depth \\after optimization
        \end{tabular} & $6.00\times 10^{12}$ & $7.55\times 10^{11}$ \\ \hline
        \begin{tabular}{l}
            \#non-Clifford gates \\before optimization
        \end{tabular} & $1.14\times 10^{13}$ & $1.43\times 10^{12}$ \\ \hline
        \begin{tabular}{l}
            \#non-Clifford gates \\after optimization
        \end{tabular} & $7.62\times 10^{12}$ & $9.59\times 10^{11}$ \\ \hline
        %量子ビット数 & 10241 & 6659 & 6659 \\ \hline
        %\begin{tabular}{c}
        %    最適化前の\\深さ
        %\end{tabular}
        %& $1.19\times 10^{13}$ & $1.51\times 10^{12}$ & $3.39\times 10^{12}$ \\ \hline
        %\begin{tabular}{c}
        %    最適化後の\\深さ
        %\end{tabular}
        %& $5.44\times 10^{12}$ & $6.86\times 10^{11}$ & $1.54\times 10^{12}$ \\ \hline
        %\begin{tabular}{c}
        %    最適化前の\\non-Clifford\\ゲート数
        %\end{tabular}
        %& $1.02\times 10^{13}$ & $1.29\times 10^{12}$ & $2.89\times 10^{12}$ \\ \hline
        %\begin{tabular}{c}
        %    最適化後の\\non-Clifford\\ゲート数
        %\end{tabular}
        %& $6.82\times 10^{12}$ & $8.60\times 10^{11}$ & $1.93\times 10^{12}$ \\ \hline
    \end{tabular}
    \label{tab:r-add-2048}
\end{table}

%Schnorr groupの強度は量子計算のもとではどこまで低下するのかを、safe-prime groupの$p$の大きさで測った。
The strength of Schnorr groups under quantum computation can be measured by the size of $p$ in safe-prime groups.
%つまり、Shorの量子アルゴリズムを用いたとき、$p=2048$ bitのSchnorr groupの計算量は、どの大きさの$p$のsafe-prime groupの計算量と一致するのかを調べた。
Specifically, when using Shor's quantum algorithm, we investigated the size of $p$ in safe-prime groups that matches the computational complexity of the Schnorr group with $p=2048$ bits.
%式\eqref{eq:fitting}において、safe-prime groupでは$\log q=\log p-1$なので、
In \eqref{eq:fitting}, for safe-prime groups, since $\log q=\log p-1$, it can be transformed as follows:
\begin{align}
    &a(\log p)^2(\log p-1) + b(\log p)^2 + c\log p(\log p-1) \notag\\
    &+ d\log p + e(\log p-1) + f \\
    =&a(\log p)^3 + (-a+b+c)(\log p)^2 + (-c+d+e)\log p \notag\\
    &+ (-e+f). \label{eq:fitting-safeprime}
\end{align}
%と変形できる。
%各リソースについてこの式\eqref{eq:fitting-safeprime}にパラメータ（表\ref{tab:fitting-parameters}）を代入したものと、表\ref{tab:r-add-2048}のSchnorr groupの結果が一致する$p$を求めると、表\ref{tab:schnorr-strength-in-safeprime}のようになる。
By substituting the parameters (Table~\ref{tab:fitting-parameters}) into this \eqref{eq:fitting-safeprime} for each resource and finding the $p$ that matches the results of the Schnorr group in Table~\ref{tab:r-add-2048}, we obtain Table~\ref{tab:schnorr-strength-in-safeprime}.
%このように、量子計算の枠組みではちょうど半分のビット数の$p$の暗号強度しかSchnorr groupはもっていないと明らかになった。
%Thus, it became clear that in the framework of quantum computation, the Schnorr group only has half the cryptographic strength of the $p$ bit length.
%ここでは最適化の有無にかかわらずゲート数・回路深さ・non-Cliffordゲート数のいずれであっても、Shorの量子アルゴリズムのもとで$p=2048$ bitのSchnorr groupはちょうどほぼ$p=1024$ bitのsafe-prime groupの暗号強度しかもっていないと示されている。
Here, it is shown that regardless of optimization, the number of gates, circuit depth, or the number of non-Clifford gates, under Shor's quantum algorithm, a Schnorr group with $p=2048$ bits has almost the same cryptographic strength as a safe-prime group with $p=1024$ bits.
%ゆえに、式\eqref{eq:fitting}のパラメータ$a$,$b$,$c$,$d$,$e$,$f$にかかわらない普遍性が存在すると推測できる。
Therefore, it can be inferred that there is a universality independent of the parameters $a$, $b$, $c$, $d$, $e$, and $f$ in \eqref{eq:fitting}.

\begin{table}
    \caption{
        %Safe-prime groupの式\eqref{eq:fitting-safeprime}に各パラメータを代入したものと表\ref{tab:r-add-2048}のSchnorr groupの$p=2048$ bitリソース推定結果が一致する$p$の大きさ
        The size of $p$ where the substitution of each parameter into the \eqref{eq:fitting-safeprime} for the safe-prime group matches the $p=2048$ bit resource estimation result of the Schnorr group in Table~\ref{tab:r-add-2048}.
    }
    \centering
    \begin{tabular}{|l|l|}
        \hline
        & $|p|$ \\ \hline
        \#gates before optimization & $1024.510$ \\ \hline
        \#gates after optimization & $1024.533$ \\ \hline
        depth before optimization & $1024.846$ \\ \hline
        depth after optimization & $1024.861$ \\ \hline
        \#non-Clifford gates before optimization & $1024.515$ \\ \hline
        \#non-Clifford gates after optimization & $1024.532$ \\ \hline
    \end{tabular}
    \label{tab:schnorr-strength-in-safeprime}
\end{table}

%今回は$p=2048$ bitのときにちょうど半分の$p=1024$ bitと実験的に示されたが、一般の$p$における挙動を確認する。
%This time, it was experimentally shown that when $p=2048$ bits, it is exactly half, $p=1024$ bits, but we will confirm the behavior for general $p$.
%そこでここからは、量子計算においてSchnorr groupとsafe-prime groupの暗号強度が等しくなる$p$について、より一般的な関係式を理論的に解析する。
Henceforth, we theoretically analyze a more general relational expression for $p$ where the cryptographic strength of the Schnorr group and the safe-prime group becomes equal in quantum computation.
%Schnorr groupにおける$p$と$q$のビット数の組合せは\cite{dss}より、$(|p|,|q|)=(1024,160),(2048,224),(2048,256),(3072,256)$であるため、およそ
According to \cite{dss}, the combination of bit lengths for $p$ and $q$ in Schnorr groups is $(|p|,|q|)=(1024,160),(2048,224),(2048,256),(3072,256)$, so approximately
\begin{align}
    \log q=\chi \log p
\end{align}
%ただし、
where
\begin{align}
    \chi\approx 1/10 \label{eq:chi}
\end{align}
%と表される。
is expressed.
%すると、式\eqref{eq:fitting}より、
Then, from \eqref{eq:fitting},
\begin{align}
    &a(\log p)^2(\chi \log p) + b(\log p)^2 + c(\log p)(\chi \log p) \notag\\
    &+ d\log p + e\chi\log p + f \\
    =&a\chi(\log p)^3 + (b+c\chi)(\log p)^2 + (d+e\chi)\log p + f \label{eq:fitting-schnorr}
\end{align}
%と変形できる。
%can be transformed.
is obtained.
%Safe-prime groupの$p$を$p_\mathrm{sp}$、Schnorr groupの$p$を$p_\mathrm{S}$として、式\eqref{eq:fitting-safeprime}と式\eqref{eq:fitting-schnorr}を等号で結んだ方程式を変形すると
Let $p_\mathrm{sp}$ be the $p$ of the safe-prime group and $p_\mathrm{S}$ be the $p$ of the Schnorr group.
By transforming the equation that connects \eqref{eq:fitting-safeprime} and \eqref{eq:fitting-schnorr} with an equal sign, we get
\begin{align}
    &a((\log p_\mathrm{sp})^3-\chi(\log p_\mathrm{S})^3) + (-a+b+c)(\log p_\mathrm{sp})^2 \notag\\
    &- (b+c\chi)(\log p_\mathrm{S})^2 + (-c+d+e)\log p_\mathrm{sp} \notag\\
    &- (d+e\chi)\log p_\mathrm{S} -e = 0.
\end{align}
%\begin{align}
%    &a((\log p_\mathrm{sp})^3-(\log p_\mathrm{sp})^2-\chi(\log p_\mathrm{S})^3) + b((\log p_\mathrm{sp})^2-(\log p_\mathrm{S})^2) \\
%    + &c((\log p_\mathrm{sp})^2-\log p_\mathrm{sp}-\chi(\log p_\mathrm{S})^2) + d(\log p_\mathrm{sp}-\log p_\mathrm{S}) \\
%    + &e(\log p_\mathrm{sp}-1-\chi\log p_\mathrm{S}) = 0
%\end{align}
%となる。
%ここで、$\log p_\mathrm{sp}$と$\log p_\mathrm{S}$の2次の項までを無視すると、
By ignoring the second-order terms of $\log p_\mathrm{sp}$ and $\log p_\mathrm{S}$, we have
\begin{align}
    (\log p_\mathrm{sp})^3 &\sim \chi(\log p_\mathrm{S})^3 \\
    \log p_\mathrm{sp} &\sim \sqrt[3]{\chi} \log p_\mathrm{S}.
\end{align}
%式\eqref{eq:chi}より、両者のビット数は2倍程度の差があるとわかる。
From \eqref{eq:chi}, it is understood that there is a difference of about twice in the bit length of both.
%特に、Schnorr groupにて$(|p|,|q|)=(2048,256)$の場合は$\chi=8=2^3$となるので、$\log p_\mathrm{sp}\sim 2\log p_\mathrm{S}$、つまりShorアルゴリズムのもとで$p=2048$ bitのSchnorr groupの暗号強度は$p=1024 bit$のsafe-primeとほぼ同等であるという表\ref{tab:schnorr-strength-in-safeprime}の結果を再現できる。
In particular, in the Schnorr group with $(|p|,|q|)=(2048,256)$, $\chi=1/8=1/2^3$, so $\log p_\mathrm{sp}\sim (1/2)\log p_\mathrm{S}$.
This means that under Shor's algorithm, the cryptographic strength of a Schnorr group with $p=2048$ bits is almost equivalent to that of a safe-prime with $p=1024$ bits, reproducing the results shown in Table~\ref{tab:schnorr-strength-in-safeprime}.
%したがって、具体的なパラメータによらずとも、Schnorr groupの暗号強度はsafe-prime groupの$p$のビット数では半分程度にまで低下すると示された。
Therefore, it was shown that the cryptographic strength of the Schnorr group is reduced to about half in terms of the bit length of $p$ in the safe-prime group, regardless of specific parameters.

%さて、表\ref{tab:r-add-2048}から貪欲法による回路最適化\cite{qulacs}前後のリソース変化量を計算できる。
%深さについてはおよそ0.45倍、non-Cliffordゲート数についてはおよそ0.67倍に削減できている。
%したがって、先行研究におけるリソース見積もりにおいても同様に回路最適化を施してリソースをさらに削減できると考えられる。
%特に、今回用いたR-ADDベースである\cite{Gidney2021howtofactorbit}において同程度のリソース削減率が適用できると仮定し、
%計算時間とnon-Cliffordゲート数の増減幅が一致しているとしたとき、
%$p=2048$ bitにおいて、安全素数型では7.0時間が $7.0\times 0.67\approx 4.7$ 時間、Schnorr型では0.8時間が $0.8\times 0.67\approx 0.54$ 時間と計算時間は削減される。
%このうち、安全素数型の計算時間はRSA型合成数の素因数分解と同程度の時間がかかる\cite{Gidney2021howtofactorbit}。
%一方、現在の量子計算機では1時間に1回程度の間隔で致命的なエラーが発生する\cite{google2024quantumerrorcorrection}ことから、Schnorr型についてはその範疇で時間削減されるため解かれる危険性が増大している。

\section{Conclusion}

%\textcolor{red}{（変更予定）}
%
%Q-ADDをベースとした量子回路を実際に量子シミュレータ上で動かし、成功確率について\cite{martin2019revisitingshor}で考えられた理論と一致することを実証した。
%また、R-ADDをベースとして生成した回路から最適化前後での量子回路の深さとnon-Cliffordゲート数の増減を推定した。
%ここから、Schnorr型はSafe-prime型と比べて量子計算機にとって解くのに必要なリソースが少なくなることを示した。
%したがって、量子計算機がより性能を向上させて大規模な論理量子演算が可能となる前に、少なくともSchnorr型の使用は取りやめるべきである。

%DLPを解くShorアルゴリズムについて、これまでおこなわれてこなかった標数$p$と位数$q$の包括的な組合せの大規模シミュレーションを実施した。
We conducted large-scale quantum simulations of comprehensive combinations of modulo $p$ and order $q$ for Shor's algorithm solving the DLP, which had not been done before.
%32量子ビットまでのすべての$p$と$q$の組合せ1860通りの実験を通して、成功確率の実際の値を一覧として得ることができた。
Through 1,860 experiments with all combinations of $p$ and $q$ up to 32 qubits, we obtained a list of actual success probabilities.
%%そこから、Ekeråのヒューリスティックな理論解析\cite{martin2019revisitingshor}の正しさを示し、そのうえで位数$q$による成功確率の周期的な変動や、最小値・最大値を発見した。
%From this, we demonstrated the correctness of Ekerå's heuristic theoretical analysis~\cite{martin2019revisitingshor}, and discovered periodic fluctuations in success probabilities due to order $q$, as well as minimum and maximum values.
%そこから、位数$q$による成功確率の周期的な変動を含むEkeråのヒューリスティックな理論解析[20]の正しさを示し、そのうえで最小値・最大値や詳細な波形の形を得られた。
From this, the correctness of Ekerå's heuristic theoretical analysis~\cite{martin2019revisitingshor}, which includes periodic fluctuations in success probability due to order $q$, was demonstrated, and the minimum and maximum values as well as the detailed waveform shape were obtained.
%また、より大きな値の$p$と$q$の組合せの実験を通して得られた回路規模を外挿し、2048 bitにおける必要なリソース値推定をした。
Additionally, by extrapolating the circuit scale obtained from experiments with larger combinations of $p$ and $q$, we estimated the necessary resource values for 2048 bits.
%ここから、Shorの量子アルゴリズムのもとでSchnorr groupはsafe-prime groupよりも弱い暗号強度であることを示し、そのうえで、Schnorr groupの暗号強度はsafe-prime groupにとっては半分程度の大きさの$p$にまで低下することを実験的・理論的に示した。
From this, we showed that under Shor's quantum algorithm, the Schnorr group has weaker cryptographic strength than the safe-prime group, and furthermore, the cryptographic strength of the Schnorr group experimentally and theoretically decreases to about half the size of $p$ for the safe-prime group.

%今後は、\cite{Gidney2021howtofactorbit,Gouzien2021factoring}のような量子エラー訂正を前提とした際や、elliptic curve discrete logarithm problemへの拡張をしたいと考えている。
In the future, we aim to extend this to scenarios assuming quantum error correction as in \cite{Gidney2021howtofactorbit,Gouzien2021factoring}, and to the elliptic curve discrete logarithm problem.

\section*{Acknowledgment}

%The authors deeply thank Martin Ekerå for
%著者らは成功確率の振動の既存研究を含め様々な有用な指摘をしてくださったMartin Ekeråに深く感謝している。
The authors are deeply grateful to Martin Ekerå for his various valuable comments, including existing research on the oscillation of success probability.
\appendix

\section{2-controlling qubit trick is optimal in time when time of modular exponentiation $>$ time of measurement}

%この節では図\ref{fig:circuit-2qubits-dlp}に関連して、冪剰余計算にかかる時間が測定時間よりも長い場合は測定量子ビット数が2であれば最小の時間で計算可能であることを示す。
In this section, in relation to Figure~\ref{fig:circuit-2qubits-dlp}, we demonstrate that if the time required for modular exponentiation is longer than the measurement time, the computation can be performed in the minimum time with two measurement qubits.
%冪剰余計算の各ユニットMod-MULにかかる時間を$U$、測定および古典的条件つき演算にかかる時間を$M$とする。
Let $U$ be the time required for each unit Mod-MUL of modular exponentiation, and $M$ be the time required for measurement and classical conditional operations.
%また、Mod-MULの計算回数を$L$とする。
Also, let $L$ be the number of Mod-MUL computations.
%これは、素因数分解向けのShorアルゴリズムであれば合成数のビット数の2倍、DLP向けのShorアルゴリズムであれば式\eqref{eq:1st-2nd-registers}となる。
For Shor's algorithm for factorization, this is twice the number of bits of the composite number, and for Shor's algorithm for DLP, it is given by \eqref{eq:1st-2nd-registers}.

\begin{theorem}
    %$U\geq M$のとき、回路全体でかかる時間は
    When $U\geq M$, the total time required for the circuit is minimized by
    \begin{align}
        LU+M\label{eq:circuit-time-u-geq-m}
    \end{align}
    %が最小であり、それを得られる最小の制御ビット数は2である。
    and the minimum number of control qubits to achieve this is 2.
\end{theorem}
\begin{proof}
    %冪剰余計算の対象となるレジスタは1つのみでありそこで$L$回のMod-MULを実行する必要があるので、計算時間の下界は$L\times U$に、その後の測定時間$M$を加えた式\eqref{eq:circuit-time-u-geq-m}となる。
    Since there is only one register for modular exponentiation and $L$ Mod-MULs need to be executed there, the lower bound of the computation time is $L\times U$, and adding the subsequent measurement time $M$ results in \eqref{eq:circuit-time-u-geq-m}.
    %また、制御量子ビット数が1つのときはMod-MULを実行するたびに、1つのみの制御量子ビットにて測定と古典的条件つき演算を実行しなければならないので、全体の時間は
    Moreover, when there is one control qubit, measurement and classical conditional operations must be executed with only one control qubit each time a Mod-MUL is executed, so the total time is
    \begin{align}
        LU+LM=L(U+M).\label{eq:circuit-time-1control}
    \end{align}
    %となる。
    %一方、制御量子ビット数が2つになると、$U\geq M$より、Mod-MULを片方の制御量子ビットで制御実行している間に、もう片方の制御量子ビットで1つ前の測定・古典的条件つき演算を実行できるので、その全体の時間は式\eqref{eq:circuit-time-u-geq-m}となる。
    On the other hand, when there are two control qubits, since $U\geq M$, while one control qubit is executing the controlled Mod-MUL, the other control qubit can execute the measurement and classical conditional operations from the previous step, so the total time is given by \eqref{eq:circuit-time-u-geq-m}.
\end{proof}
\section{Decision of \lowercase{$k$} such that \lowercase{$k$}-controlling qubit trick is optimal in time when $U<M$}
%
%\textcolor{red}{(TODO: the case of $U<M$ as in my patent)}

%前セクションでは$U\geq M$の場合を扱ったが、$U<M$のときは最適な方法が異なる。
In the previous section, we dealt with the case where $U\geq M$, but when $U<M$, the optimal method differs.
%制御量子ビット数を$k\ll L$とする。
Let the number of control qubits be $k\ll L$.
%$k=1$のときに回路全体でかかる時間は式\eqref{eq:circuit-time-1control}である一方、$k\geq 2$では別の制御量子ビットで測定・古典的条件つき演算を実行する際にそれまでの測定結果がすべて必要となることから並列実行ができず、結果として
When $k=1$, the time taken for the entire cicuit is given by \eqref{eq:circuit-time-1control}, whereas for $k\geq 2$, all previous measurement results are required for executing measurements and classically conditioned operations with other control qubits, preventing parallel execution.
As a result, it takes a time of
\begin{align}
    U+LM.
\end{align}
%の時間がかかる。

%そこで、従来のように制御量子ビット1つずつで半古典逆量子フーリエ変換をするのではなく、図のように$k$量子ビットでまとめたso-called 半々古典逆量子フーリエ変換を用いる。
Therefore, instead of performing the semi-classical inverse quantum Fourier transform one control qubit at a time as traditionally done, we use the so-called half-semi-classical inverse quantum Fourier transform with $k$ qubits as shown in Figure~\ref{fig:half-semi-qft}.
%ここでは$k$量子ビットで並列に測定を実行できるので、$U<M$のときにはより高速となることが期待できる。
Here, measurements can be executed in parallel with $k$ qubits, so it is expected to be faster when $U<M$.
%この半々古典逆量子フーリエ変換にかかる時間は、$U<M$では$M$とみなすことができる。
The time taken for this half-semi-classical inverse quantum Fourier transform can be considered as $M$ when $U<M$.
%すると、回路全体でかかる時間は
Thus, the time taken for the entire circuit is
\begin{align}
    LU+\frac{L}{k}M=L\left(U+\frac{1}{k}M\right). \label{eq:circuit-time-hsqft}
\end{align}
%となる。

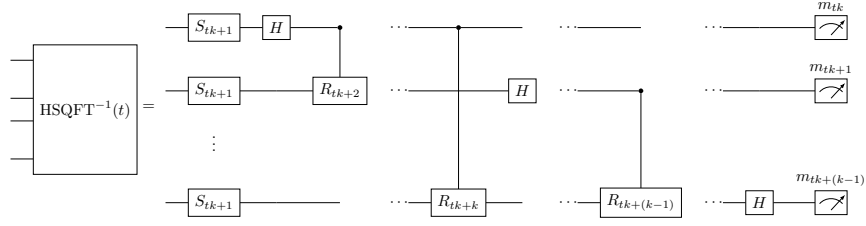
\begin{figure*}
    \centering
    \begin{tikzpicture}
        \node[scale=0.6]{
\begin{quantikz}[thin lines]
    & \gate[4]{\mathrm{HSQFT}^{-1}(t)} \\
    & \\
    & \\
    &
\end{quantikz}=\begin{quantikz}[thin lines]
    & \gate{S_{tk+1}} & \gate{H} & \ctrl{1}        & \cdots & \ctrl{3}        & \qw      & \cdots & \qw                 & \cdots & \qw      & \meter{$m_{tk}$} \\
    & \gate{S_{tk+1}} & \qw      & \gate{R_{tk+2}} & \cdots & \qw             & \gate{H} & \cdots & \ctrl{2}            & \cdots & \qw      & \meter{$m_{tk+1}$} \\
    & \vdots                                                                                                                                                  \\
    & \gate{S_{tk+1}} & \qw      & \qw             & \cdots & \gate{R_{tk+k}} & \qw      & \cdots & \gate{R_{tk+(k-1)}} & \cdots & \gate{H} & \meter{$m_{tk+(k-1)}$}
\end{quantikz}
        };
    \end{tikzpicture}
    \caption{
        %半々古典逆量子フーリエ変換。$R_j$および$S_j$の定義は図\ref{fig:circuit-standard-dlp}と\ref{fig:circuit-1qubit-dlp}と同じである。
        Half-semi-classical inverse quantum Fourier transform.
        The definitions of $R_j$ and $S_j$ are the same as in Figure~\ref{fig:circuit-standard-dlp} and \ref{fig:circuit-1qubit-dlp}.
    }
    \label{fig:half-semi-qft}
\end{figure*}

\begin{figure*}
    \centering
    \begin{tikzpicture}
        \node[scale=0.6]{
\begin{quantikz}
    \lstick{$\ket{+}$}          & \ctrl{3}                  & \qw                       & \qw                       & \gate[3]{\mathrm{HSQFT}^{-1}(0)}&     & \lstick{$\ket{+}$} & \ctrl{3}                  & \qw                       & \qw                       & \gate[3]{\mathrm{HSQFT}^{-1}(1)} & \cdots \\
    \lstick{$\ket{+}$}          & \qw                       & \ctrl{2}                  & \qw                       &                                 &     & \lstick{$\ket{+}$} & \qw                       & \ctrl{2}                  & \qw                       &                                  & \cdots \\
    \lstick{$\ket{+}$}          & \qw                       & \qw                       & \ctrl{1}                  &                                 &     & \lstick{$\ket{+}$} & \qw                       & \qw                       & \ctrl{1}                  &                                  & \cdots \\
    \lstick{$\ket{0\cdots 01}$} & \gate{h^{2^{L-1}}\bmod p} & \gate{h^{2^{L-2}}\bmod p} & \gate{h^{2^{L-3}}\bmod p} & \qw                             & \qw & \qw                & \gate{h^{2^{L-4}}\bmod p} & \gate{h^{2^{L-5}}\bmod p} & \gate{h^{2^{L-6}}\bmod p} & \qw                              & \cdots
\end{quantikz}
        };
    \end{tikzpicture}
    \caption{
        %半々古典逆量子フーリエ変換（図\ref{fig:half-semi-qft}）を用いた$k=3$のときの量子回路。
        Quantum cicuit for $k=3$ using the half-semi-classical inverse quantum Fourier transform (Figure~\ref{fig:half-semi-qft}).
    }
    \label{fig:use-half-semi-qft}
\end{figure*}
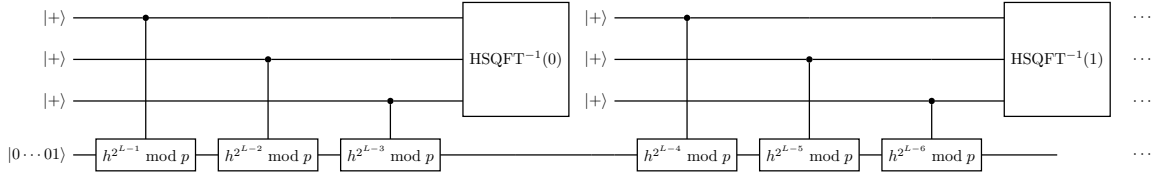

Let
\begin{align}
    U=\rho M \label{eq:u-rho-m}
\end{align}
where $(0<\rho<1)$,
Then, the following theorem holds.
%とすると、以下の定理が成り立つ。%半古典逆量子フーリエ変換を利用したときよりも回路全体でかかる時間は次の$k$で式\eqref{eq:circuit-time-1control}よりも小さくなる。
%回路全体でかかる時間が半古典逆量子フーリエ変換を利用した式\eqref{eq:circuit-time-1control}よりも小さくなる
\begin{theorem}
    %$U<M$のとき、半々古典逆量子フーリエ変換を利用することで回路全体でかかる時間が半古典逆量子フーリエ変換を利用したときよりも小さくなる$k$は
    When $U<M$, the value of $k$ for which the total time taken by the circuit using the half-smi-classical inverse quantum Fourier transform is less than when using the semi-classical inverse quantum Fourier transform is given by
    \begin{align}
        k>\frac{L}{L+\rho-L\rho}. \label{eq:minimum-k-u-l-m}
    \end{align}
    %を満たせば良い。
\end{theorem}
\begin{proof}
    %式\eqref{eq:circuit-time-1control}と式\eqref{eq:circuit-time-hsqft}を比較すると
    By comparing \eqref{eq:circuit-time-1control} and \eqref{eq:circuit-time-hsqft}, we have
    \begin{align}
        L\left(U+\frac{1}{k}M\right) < U+LM.
    \end{align}
    %式\eqref{eq:u-rho-m}を代入して変形すると、式\eqref{eq:minimum-k-u-l-m}となる。
    Substituting \eqref{eq:u-rho-m} and rearranging, we obtain \eqref{eq:minimum-k-u-l-m}.
\end{proof}

%前節とこの節を合わせることで、それぞれの場合における回路全体でかかる時間を最小にした方法で実装ができる。
By combining the previous section with this one, the total time taken by the circuit can be minimized in any case.

\section{Using another control register as clean qubits in R-ADD}

%ゲート数において、素因数分解向けのShorアルゴリズムと異なる減らし方がR-ADDにおいて可能である。
In terms of gate count, a different reduction method from Shor's algorithm for factoring is possible in R-ADD.
%R-ADDでは3--controlled NOT and 4--controlled NOTを用いるが、その構築に使う補助ビットがcleanかdirtyかで必要なゲート数が変化する。
R-ADD uses 3--controlled NOT and 4--controlled NOT gates, and the number of gates required varies depending on whether the auxiliary bits used for their construction are clean or dirty.
%\cite{yamaguchi2023experimentsresource,barenco1995elementarygates}で示されているように具体的には、3--controlled NOT gateにおいて1 clean auxiliary qubitを使うと3 Toffoli gatesで済む一方、1 dirty auxiliary qubitを使うと4 Toffoli gates必要となる。
Specifically, as shown in \cite{yamaguchi2023experimentsresource,barenco1995elementarygates}, a 3--controlled NOT gate requires 3 Toffoli gates when using 1 clean auxiliary qubit, while it requires 4 Toffoli gates when using 1 dirty auxiliary qubit.
%また、4--controlled NOT gateでは2 clean auxiliary qubitsを使うと5 Toffoli gates, 1 clean auxiliary qubitと1 dirty auxiliary qubitを使うと6 Toffoli gates, そして2 dirty auxiliary qubitsを使うと8 Toffoli gatesが必要となる。
For a 4--controlled NOT gate, 5 Toffoli gates are needed with 2 clean auxiliary qubits, 6 Toffoli gates with 1 clean and 1 dirty auxiliary qubit, and 8 Toffoli gates with 2 dirty auxiliary qubits.
%素因数分解のShorアルゴリズム\cite{yamaguchi,yamaguchi2023experimentsresource}ではclean auxiliary qubitsとして、制御レジスタの量子ビットをそれぞれの冪剰余計算モジュールを実行する前に使うことができるが、
%最後から2番目では1つはdirtyとなり、最後では2つともdirtyとなる。
In Shor's algorithm for factoring~\cite{yamaguchi,yamaguchi2023experimentsresource}, the qubits of the control register can be used as clean auxiliary qubits before executing each modular exponentiation module, but one becomes dirty in the second-to-last module, and both become dirty in the last module.
%一方、今回の離散対数問題のShorアルゴリズムでは制御レジスタが2つ存在するので、最初の制御レジスタから冪剰余計算をする場合はもう片方の制御レジスタを構成する量子ビットをclean qubitsとして使うことができる。
On the other hand, in Shor's algorithm for the DLP, there are two control registers, so when performing modular exponentiation from the first control register, the qubits constituting the other control register can be used as clean qubits.

\section{DLP Pattern Requiring More Gates Than Other Patterns}

%\textcolor{red}{(TODO: pattern 3 in ISEC version only targeting w/o semi-classical qft))}
%%\begin{enumerate}
%%    \item 古典手法と異なり、時間計算量が $p$ と $q$ の両方含む形で表される
%%\end{enumerate}
%
%%\textcolor{red}{（この節は丸ごと削除、もしくはAppendixへ移動予定）}
%
%%本節では，量子計算機において量子ビット数を固定したときに、他の離散対数問題よりも解くのにゲート数が多くなるパターンを導入する．
%本節では，半古典逆量子フーリエ変換や半々古典逆量子フーリエ変換を利用しないときに限定するものの、量子計算機において量子ビット数を固定したときに、他の離散対数問題よりも解くのにゲート数が多くなる群を導入する．
In this section, we introduce a group that requires more gates to solve than other DLP groups when the number of qubits is fixed in a quantum computer, although we limit ourselves to cases where neither the semi-classical inverse quantum Fourier transform nor the half-semi-classical inverse quantum Fourier transform is used.
%%このパターンを他のパターンに倣ってパターン3とする．
%%このパターンを\textit{Same~\#qubits}型とする。
%このパターンを\textit{Same-\#qubits group}とする。
This group is referred to as the \textit{Same-\#qubits group}.
%量子ビット数（式\eqref{eq:general-space-complexity}）を固定したときに必要なゲート数（式\eqref{eq:general-time-complexity}）が最も多くなる$p$, $q$の組合せを求める．
We seek the combination of $p$ and $q$ that requires the most gates \eqref{eq:general-time-complexity} when the number of qubits \eqref{eq:general-space-complexity} is fixed.
%この問題は量子ビット数についてより一般化すると，次のように定義できる．
This problem can be generalized with respect to the number of qubits and defined as follows:
%%この問題は今回用いた冪剰余のオラクルに限らずに一般化すると，次のように定義できる．
\begin{definition}[
    %量子ビット数を固定したゲート数の最大化問題
    Maximization problem of gate count with fixed qubit number
    ]
    \begin{align}
        \alpha x+\beta y-\gamma=0\label{eq:constraint}
    \end{align}
    %の制約条件つきの
    subject to the constraint
    \begin{align}
        x^m y^n
    \end{align}
    %の最大化問題．
    maximization problem.
    %ただし，$\alpha,\beta,\gamma,m,n,x,y>0$．
    Here, $\alpha,\beta,\gamma,m,n,x,y>0$.
\end{definition}

%これは制約条件（式\eqref{eq:constraint}）を$y=(\gamma-\alpha x)/\beta$（$y>0$より$x<\gamma/\alpha$）と変形してから
This can be solved by transforming the constraint \eqref{eq:constraint} to $y=(\gamma-\alpha x)/\beta$ (since $y>0$, $x<\gamma/\alpha$) and defining
\begin{align}
    F(x)=x^m y^n=x^m\left(\frac{\gamma-\alpha x}{\beta}\right)^n
\end{align}
%と定義し，この微分を利用して増減表（表\ref{tab:first-derivative-test}）を書くと解くことができる．
and using its derivative to write a first derivative test (Table~\ref{tab:first-derivative-test}).
%この結果より，
From this result,
\begin{align}
    \lceil\log p\rceil&=x=\frac{\gamma m}{\alpha(m+n)}\label{eq:x-max}\\
    \lfloor\log q\rfloor&=y=\frac{\gamma n}{\beta(m+n)}\label{eq:y-max}
\end{align}
%であるときに$F(x)$が最大となり，つまり最も多いゲート数が量子計算機にて必要となる．
is when $F(x)$ is maximized, meaning the most gates are required on a quantum computer.

\begin{table}[tb]
    \centering
    \caption{
        %最大化問題を解くための増減表．
        A first derivative test for solving the maximization problem.
    }
    %\begin{center}
        \begin{tabular}{|c|c|c|c|c|c|}
            \hline
            $x$ & $0$ & $\cdots$ & $\frac{\gamma m}{\alpha(m+n)}$ & $\cdots$ & $\frac{\gamma}{\alpha}$ \\ \hline
            $F'(x)$ & $0$ & $+$ & $0$ & $-$ & \\ \hline
            $F(x)$ & $0$ & $\nearrow$ & $\frac{\gamma^{m+n}m^m n^n}{\alpha^m \beta^n (m+n)^{m+n}}$ & $\searrow$ & $0$ \\
            \hline
        \end{tabular}
    %\end{center}
    \label{tab:first-derivative-test}
\end{table}

%これは$x=\lceil\log p\rceil$と$y=\lfloor\log q\rfloor$の間の関係式として書き換えることができる．
This can be rewritten as a relational expression between $x=\lceil\log p\rceil$ and $y=\lfloor\log q\rfloor$:
\begin{align}
    \lceil\log p\rceil=\frac{\beta m}{\alpha n}\lfloor\log q\rfloor.\label{eq:pattern3-logp-logq-relation}
\end{align}
%$p,q\gg 1$のときは$\lceil\log p\rceil\approx\log p$, $\lfloor\log q\rfloor\approx\log q$であるので
When $p,q\gg 1$, $\lceil\log p\rceil\approx\log p$ and $\lfloor\log q\rfloor\approx\log q$, so
\begin{align}
    \log p&\approx\frac{\beta m}{\alpha n}\log q \\
    p&\approx q^\frac{\beta m}{\alpha n}\label{eq:difficult-p-q-in-quantum}
\end{align}
%と表される．
is expressed.

%これは冪剰余の具体的な計算と組み合わせることでシンプルな形式となる．
This becomes a simple form when combined with specific calculations of modular exponentiations.
%例えばQ-ADDを利用すると，空間計算量は式\eqref{eq:qadd-space-complexity}を式\eqref{eq:general-space-complexity}に代入したもの、時間計算量は式\eqref{eq:qadd-time-complexity}である
%ため，$\alpha=2$, $\beta=2$, $\gamma=\qubits-4$, $m=3$, $n=1$となる．このうち$\gamma$は式\eqref{eq:difficult-p-q-in-quantum}に現れないので，
%\begin{align}
%    p\approx q^3\label{eq:qadd-pattern3}
%\end{align}
%例えばR-ADDを利用すると，空間計算量は式\eqref{eq:radd-space-complexity}を式\eqref{eq:general-space-complexity}に代入したもの、時間計算量は式\eqref{eq:radd-time-complexity}である
For example, using R-ADD, the space complexity is the substitution of \eqref{eq:radd-space-complexity} into \eqref{eq:general-space-complexity}, and the time complexity is \eqref{eq:radd-time-complexity},
%ため，$\alpha=3$, $\beta=2$, $\gamma=\qubits-4$, $m=2$, $n=1$となる．このうち$\gamma$は式\eqref{eq:difficult-p-q-in-quantum}に現れないので，
so $\alpha=3$, $\beta=2$, $\gamma=\qubits-4$, $m=2$, $n=1$.
Among these, $\gamma$ does not appear in \eqref{eq:difficult-p-q-in-quantum}, so
\begin{align}
    p\approx q^{4/3}\label{eq:radd-pattern3}
\end{align}
%となる．
is obtained.
%同様に、Q-ADDの場合には
Similarly, in the case of Q-ADD, we have
\begin{align}
    p\approx q^3.\label{eq:qadd-pattern3}
\end{align}
%となる。

%Simulating Quantum Circuitsでは、時間計算量についてSchnorr groupと同一になるようなsafe-prime groupの$p$を求めたが、そこでは空間計算量については考慮していなかった。
In Simulating Quantum Circuits, $p$ of a safe-prime group was sought to match the time complexity of the Schnorr group, but space complexity was not considered there.
%このSame-\#qubits groupを用いることで、半古典逆量子フーリエ変換を利用しない場合には限定するものの、Schnorr groupと同一の空間計算量かつ、そのときに可能な最大の時間計算量をもつ$p$と$q$の組合せを求められる。
By using the same-\#qubits group, although limited to cases where the semi-classical inverse quantum Fourier transforms are not used, it is possible to find a combination of $p$ and $q$ with the same space complexity as the Schnorr group and the maximum possible time complexity.
%表\ref{tab:r-add-2048}と同様に式\eqref{eq:fitting}を用いたときの結果は表のようになる。
The results when using \eqref{eq:fitting} are shown in the Table~\ref{tab:r-add-same-qubits}, similar to Table~\ref{tab:r-add-2048}.
%表\ref{tab:schnorr-strength-in-safeprime}におけるsafe-primeほどではないものの、小さな$p$でSchnorr groupのリソース消費量を上回ることができている。
Although not as strong as the safe-prime group in Table~\ref{tab:schnorr-strength-in-safeprime}, it can exceed the resource consumption of the Schnorr group with a smaller $p$.

\begin{table}
    \centering
    \caption{
        %R-ADD~\cite{radd}を冪剰余計算に利用した際の，$p$が2048ビットのときのSchnorr groupと、その量子ビット数でのsame-\#qubits groupにおけるリソース推定結果
        Resource estimation results for the Schnorr group when $p$ is 2048 bits and the same-\#qubits group at the number of qubits of the Schnorr group, using R-ADD~\cite{radd} for modular exponentiation calculations.
    }
    \begin{tabular}{|l|l|l|}
        \hline
        ($|p|,|q|$) &
        \begin{tabular}{c}
            Schnorr\\(2048, 256)
        \end{tabular}
        &
        \begin{tabular}{c}
            Same-\#qubits\\(1479, 1109)
        \end{tabular}
        \\ \hline
        \begin{tabular}{l}
            \#qubits
        \end{tabular} & 6659 & 6659 \\ \hline
        \begin{tabular}{l}
            \#gates \\before optimization
        \end{tabular} & $3.09\times 10^{12}$ & $6.96\times 10^{12}$ \\ \hline
        \begin{tabular}{l}
            \#gates \\after optimization
        \end{tabular} & $1.00\times 10^{12}$ & $2.26\times 10^{12}$ \\ \hline
        \begin{tabular}{l}
            depth \\before optimization
        \end{tabular} & $1.67\times 10^{12}$ & $3.75\times 10^{12}$ \\ \hline
        \begin{tabular}{l}
            depth \\after optimization
        \end{tabular} & $7.55\times 10^{11}$ & $1.70\times 10^{12}$ \\ \hline
        \begin{tabular}{l}
            \#non-Clifford gates \\before optimization
        \end{tabular} & $1.43\times 10^{12}$ & $3.23\times 10^{12}$ \\ \hline
        \begin{tabular}{l}
            \#non-Clifford gates \\after optimization
        \end{tabular} & $9.59\times 10^{11}$ & $2.16\times 10^{12}$ \\ \hline
    \end{tabular}
    \label{tab:r-add-same-qubits}
\end{table}

%これはSafe-prime型やSchnorr型と比べて異なる$p$と$q$の関係式であり、それらの型が量子計算機にとってどれほど難しいかを示す指標になりうる。

%なお、Q-ADDを用いた36量子ビット回路実行実験では具体例として$p=4093$, $q=31$つまり$p$が12ビット，$q$が5ビットのSame-\#qubits groupを実験している。
In the experiment executing a 36 qubits circuit using Q-ADD, a specific example of $p=4093$, $q=31$ was tested, meaning $p$ is 12 bits and $q$ is 5 bits in the same-\#qubits group.
%その値は式\eqref{eq:qadd-pattern3}に従っていないように見えるが，これはまだ$p,q\gg 1$といえず，$\lceil\log p\rceil\not\approx \log p$および$\lfloor\log q\rfloor\not\approx \log q$となるためである．
This values do not seem to follow \eqref{eq:qadd-pattern3} because $p,q\gg 1$ is not the case, leading to $\lceil\log p\rceil\not\approx \log p$ and $\lfloor\log q\rfloor\not\approx \log q$.
%そのシミュレーションは30分で完了し、他のサンプルと同様に図\ref{fig:success-prob-heatmap},\ref{fig:success-prob-boxplot},\ref{fig:success-prob-boxplot_p}に表れている。
The simulation was completed in 30 minutes and is represented in Figures~\ref{fig:success-prob-boxplot}, \ref{fig:success-prob-boxplot_p}, and \ref{fig:success-prob-heatmap}, similar to other samples.

\begin{figure*}[h]
    \includegraphics[scale=0.5]{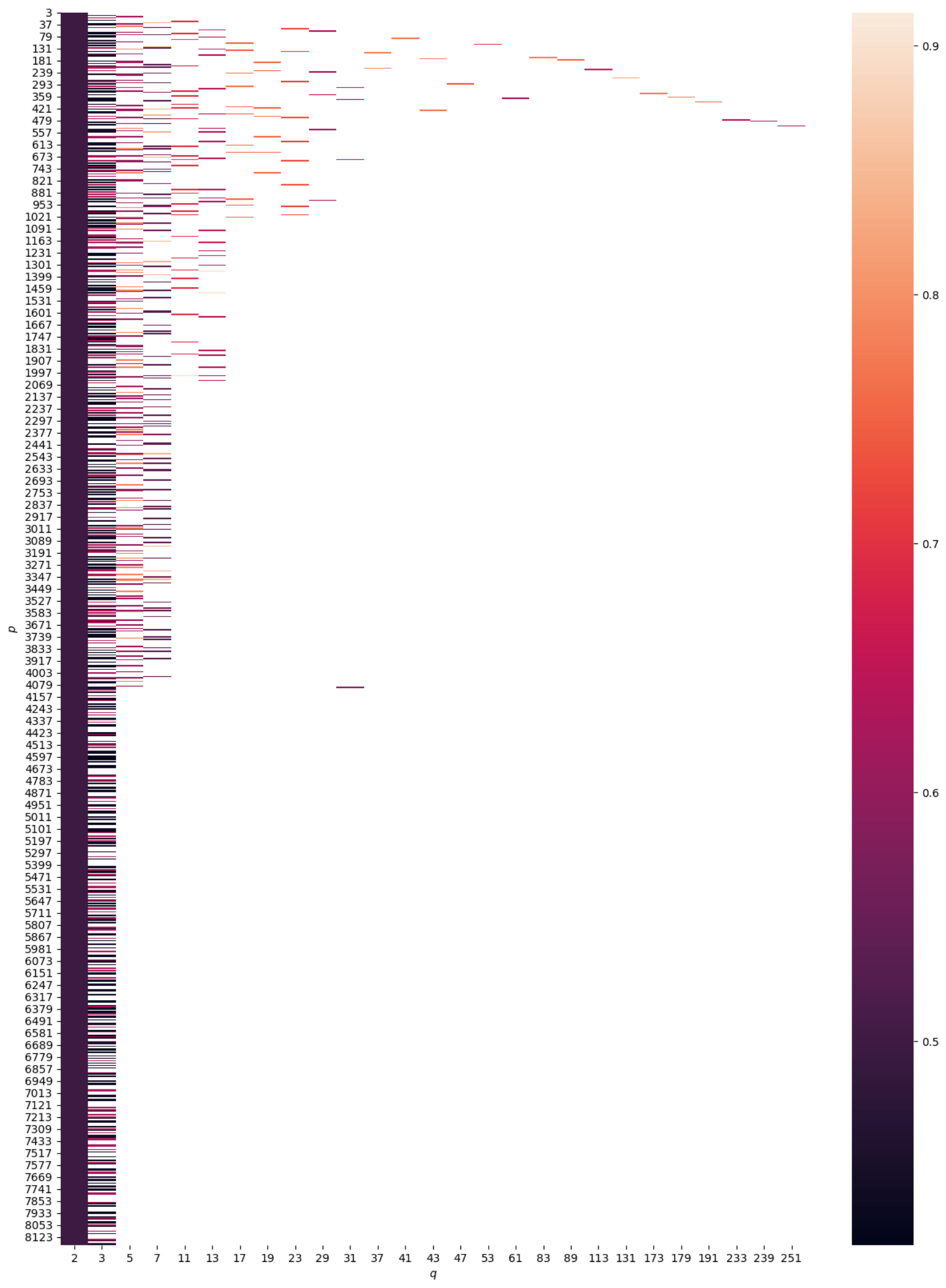}
    \caption{
        %各 $p$, $q$ の組合せにおける成功確率のヒートマップ。明るい色ほど成功確率が1に近づく。
        Heatmap of success probabilities for each combination of $p$ and $q$. Brighter colors indicate success probabilities closer to 1.
    }
    \label{fig:success-prob-heatmap}
\end{figure*}

%\begin{acknowledgments}
%\end{acknowledgments}

\bibliography{bibliography}

\begin{thebibliography}{10}

\bibitem{diffie-hellman}
W.~Diffie and M.~Hellman, ``New directions in cryptography,'' {\em IEEE
  Transactions on Information Theory}, vol.~22, no.~6, pp.~644--654, 1976.

\bibitem{dss}
{National Institute of Standards and Technology (NIST)}, ``Digital signature
  standard (dss) federal information processing standards publication 186-4.''

\bibitem{nfs}
O.~Schirokauer, ``Discrete logarithms and local units,'' {\em Philosophical
  Transactions of the Royal Society of London. Series A: Physical and
  Engineering Sciences}, vol.~345, pp.~409 -- 423, 1993.

\bibitem{barker2018nistrecommendation}
E.~Barker, L.~Chen, A.~Roginsky, A.~Vassilev, and R.~Davis, ``Recommendation
  for pair-wise key-establishment schemes using discrete logarithm
  cryptography.'' NIST SP 800-56A Rev. 3, 2018.

\bibitem{pohlig-hellman}
S.~Pohlig and M.~Hellman, ``An improved algorithm for computing logarithms
  overgf(p)and its cryptographic significance (corresp.),'' {\em IEEE
  Transactions on Information Theory}, vol.~24, no.~1, pp.~106--110, 1978.

\bibitem{baby-step}
D.~Shanks, ``Class number, a theory of factorization, and genera,'' vol.~20,
  pp.~415--440, 1971.

\bibitem{dss-fips186-5}
{National Institute of Standards and Technology (NIST)}, ``Digital signature
  standard (dss) federal information processing standards publication 186-5.''

\bibitem{shor-alt}
P.~W. Shor, ``Algorithms for quantum computation: discrete logarithms and
  factoring,'' {\em Proceedings 35th Annual Symposium on Foundations of
  Computer Science}, pp.~124--134, 1994.

\bibitem{shor}
P.~W. Shor, ``Polynomial-time algorithms for prime factorization and discrete
  logarithms on a quantum computer,'' {\em SIAM J. Comput.}, vol.~26,
  pp.~1484--1509, 1995.

\bibitem{Gidney2021howtofactorbit}
C.~Gidney and M.~Eker{\aa{}}, ``How to factor 2048 bit {RSA} integers in 8
  hours using 20 million noisy qubits,'' {\em {Quantum}}, vol.~5, p.~433, Apr.
  2021.

\bibitem{Gouzien2021factoring}
E.~Gouzien and N.~Sangouard, ``Factoring 2048-bit rsa integers in 177 days with
  13 436 qubits and a multimode memory,'' {\em Phys. Rev. Lett.}, vol.~127,
  p.~140503, Sep 2021.

\bibitem{yamaguchi}
J.~Yamaguchi, M.~Yamazaki, A.~Tabuchi, T.~Honda, T.~Izu, and N.~Kunihiro,
  ``Estimation of shor's circuit for 2048-bit integers based on quantum
  simulator.'' Cryptology ePrint Archive, Paper 2023/092, 2023.

\bibitem{akahoshi2023partiallyftqc}
Y.~Akahoshi, K.~Maruyama, H.~Oshima, S.~Sato, and K.~Fujii, ``Partially
  fault-tolerant quantum computing architecture with error-corrected clifford
  gates and space-time efficient analog rotations,'' {\em PRX Quantum}, vol.~5,
  p.~010337, Mar 2024.

\bibitem{toshio2024practicalquantumadvantage}
R.~Toshio, Y.~Akahoshi, J.~Fujisaki, H.~Oshima, S.~Sato, and K.~Fujii,
  ``Practical quantum advantage on partially fault-tolerant quantum computer,''
  2024.

\bibitem{akahoshi2024compilationoftrotter}
Y.~Akahoshi, R.~Toshio, J.~Fujisaki, H.~Oshima, S.~Sato, and K.~Fujii,
  ``Compilation of trotter-based time evolution for partially fault-tolerant
  quantum computing architecture,'' 2024.

\bibitem{Itogawa2024zerolevel}
T.~Itogawa, Y.~Takada, Y.~Hirano, and K.~Fujii, ``Even more efficient magic
  state distillation by zero-level distillation,'' 2024.

\bibitem{gidney2024msc}
C.~Gidney, N.~Shutty, and C.~Jones, ``Magic state cultivation: growing t states
  as cheap as cnot gates,'' 2024.

\bibitem{yamaguchi2023experimentsresource}
J.~Yamaguchi, M.~Yamazaki, A.~Tabuchi, T.~Honda, T.~Izu, and N.~Kunihiro,
  ``Experiments and resource analysis of shor's factorization using
  a quantum simulator,'' in {\em Information Security and Cryptology -- ICISC
  2023} (H.~Seo and S.~Kim, eds.), (Singapore), pp.~119--139, Springer Nature
  Singapore, 2024.

\bibitem{ekera2021quantumalgorithmsgeneraldlp}
M.~Ekerå, ``Quantum algorithms for computing general discrete logarithms and
  orders with tradeoffs,'' {\em Journal of Mathematical Cryptology}, vol.~15,
  no.~1, pp.~359--407, 2021.

\bibitem{ekera2017quantumalgorithmshortdlp}
M.~Eker{\aa} and J.~H{\aa}stad, ``Quantum algorithms for computing short
  discrete logarithms and factoring rsa integers,'' in {\em Post-Quantum
  Cryptography} (T.~Lange and T.~Takagi, eds.), (Cham), pp.~347--363, Springer
  International Publishing, 2017.

\bibitem{hhan2024quantumcomplexityfordlp}
M.~Hhan, T.~Yamakawa, and A.~Yun, ``Quantum complexity for discrete logarithms
  and related problems,'' in {\em Advances in Cryptology -- CRYPTO 2024}
  (L.~Reyzin and D.~Stebila, eds.), (Cham), pp.~3--36, Springer Nature
  Switzerland, 2024.

\bibitem{martin2019revisitingshor}
M.~Ekerå, ``Revisiting shor's quantum algorithm for computing general discrete
  logarithms,'' 2019.

\bibitem{mandl2022implementationsshordlp}
A.~Mandl and U.~Egly, ``Implementations for shor's algorithm for the dlp.''
  INFORMATIK 2022, 2022.

\bibitem{aono}
Y.~Aono, S.~Liu, T.~Tanaka, S.~Uno, R.~V. Meter, N.~Shinohara, and R.~Nojima,
  ``The present and future of discrete logarithm problems on noisy quantum
  computers,'' {\em IEEE Transactions on Quantum Engineering}, vol.~3,
  pp.~1--21, 2022.

\bibitem{regev2023efficientquantumfactoring}
O.~Regev, ``An efficient quantum factoring algorithm,'' 2023.

\bibitem{ekera2024extendingregevsfactoringtodlp}
M.~Eker{\aa} and J.~G{\"a}rtner, ``Extending regev's factoring algorithm
  to compute discrete logarithms,'' in {\em Post-Quantum Cryptography} (M.-J.
  Saarinen and D.~Smith-Tone, eds.), (Cham), pp.~211--242, Springer Nature
  Switzerland, 2024.

\bibitem{qadd}
S.~Beauregard, ``Circuit for shor's algorithm using 2n+3 qubits,'' {\em Quantum
  Info. Comput.}, vol.~3, p.~175–185, 3 2003.

\bibitem{radd}
V.~Vedral, A.~Barenco, and A.~Ekert, ``Quantum networks for elementary
  arithmetic operations,'' {\em Phys. Rev. A}, vol.~54, pp.~147--153, 7 1996.

\bibitem{rivest2001arestrongprimes}
R.~Rivest and R.~Silverman, ``Are 'strong' primes needed for {RSA}.''
  Cryptology {ePrint} Archive, Paper 2001/007, 2001.

\bibitem{Jou05}
A.~Joux and R.~Lercier, ``Discrete logarithms in gf(p) -- 130 digits.'' an
  email announcement, June 2005.

\bibitem{Dor06}
A.~Dorofeev, D.~Dygin, and D.~Matyukhin, ``Discrete logarithms in gf(p) -- 135
  digits.'' an email announcement, Dec. 2006.

\bibitem{Kle07}
T.~Kleinjung, ``Discrete logarithms in gf(p) -- 160 digits.'' an email
  announcement, Feb. 2007.

\bibitem{BGI+14}
C.~Bouvier, P.~Gaudry, L.~Imbert, H.~Jeljeli, and E.~Thomé, ``Discrete
  logarithms in gf(p) -- 180 digits.'' an email announcement, June 2014.

\bibitem{Kle16}
T.~Kleinjung, C.~Diem, A.~K. Lenstra, C.~Priplata, and C.~Stahlke, ``Discrete
  logarithms in gf(p) -- 768 bits.'' an email announcement, June 2016.

\bibitem{Tho19}
F.~Boudot, P.~Gaudry, A.~Guillevic, N.~Heninger, E.~Thomé, and P.~Zimmermann,
  ``795-bit factoring and discrete logarithms.'' an email announcement, Dec.
  2019.

\bibitem{BGG+20a}
F.~Boudot, P.~Gaudry, A.~Guillevic, N.~Heninger, E.~Thomé, and P.~Zimmermann,
  ``Comparing the difficulty of factorization and discrete logarithm: a
  240-digit experiment.'' Cryptology ePrint Archive, Paper 2020/697, June 2020.

\bibitem{BGG+20b}
F.~Boudot, P.~Gaudry, A.~Guillevic, N.~Heninger, E.~Thom\'{e}, and
  P.~Zimmermann, ``Comparing the difficulty of factorization and discrete
  logarithm: A 240-digit experiment,'' in {\em Advances in Cryptology –
  CRYPTO 2020: 40th Annual International Cryptology Conference, CRYPTO 2020,
  Santa Barbara, CA, USA, August 17–21, 2020, Proceedings, Part II}, (Berlin,
  Heidelberg), p.~62–91, Springer-Verlag, 2020.

\bibitem{qcqi}
M.~A. Nielsen and I.~L. Chuang, {\em Quantum Computation and Quantum
  Information: 10th Anniversary Edition}.
\newblock USA: Cambridge University Press, 10th~ed., 2011.

\bibitem{qulacs}
Y.~Suzuki, Y.~Kawase, Y.~Masumura, Y.~Hiraga, M.~Nakadai, J.~Chen, K.~M.
  Nakanishi, K.~Mitarai, R.~Imai, S.~Tamiya, T.~Yamamoto, T.~Yan, T.~Kawakubo,
  Y.~O. Nakagawa, Y.~Ibe, Y.~Zhang, H.~Yamashita, H.~Yoshimura, A.~Hayashi, and
  K.~Fujii, ``Qulacs: a fast and versatile quantum circuit simulator for
  research purpose,'' {\em {Quantum}}, vol.~5, p.~559, Oct. 2021.

\bibitem{mpiQulacs}
S.~Imamura, M.~Yamazaki, T.~Honda, A.~Kasagi, A.~Tabuchi, H.~Nakao,
  N.~Fukumoto, and K.~Nakashima, ``mpiqulacs: A distributed quantum computer
  simulator for a64fx-based cluster systems,'' 2022.

\bibitem{kunihiro}
N.~Kunihiro, ``Exact analyses of computational time for factoring in quantum
  computers,'' {\em IEICE Trans. Fundam. Electron. Commun. Comput. Sci.},
  vol.~88-A, pp.~105--111, 2005.

\bibitem{thomas2000addition}
T.~G. Draper, ``Addition on a quantum computer,'' 2000.

\bibitem{scipy}
P.~Virtanen, R.~Gommers, T.~E. Oliphant, M.~Haberland, T.~Reddy, D.~Cournapeau,
  E.~Burovski, P.~Peterson, W.~Weckesser, J.~Bright, S.~J. {van der Walt},
  M.~Brett, J.~Wilson, K.~J. Millman, N.~Mayorov, A.~R.~J. Nelson, E.~Jones,
  R.~Kern, E.~Larson, C.~J. Carey, {\.I}.~Polat, Y.~Feng, E.~W. Moore,
  J.~{VanderPlas}, D.~Laxalde, J.~Perktold, R.~Cimrman, I.~Henriksen, E.~A.
  Quintero, C.~R. Harris, A.~M. Archibald, A.~H. Ribeiro, F.~Pedregosa, P.~{van
  Mulbregt}, and {SciPy 1.0 Contributors}, ``{{SciPy} 1.0: Fundamental
  Algorithms for Scientific Computing in Python},'' {\em Nature Methods},
  vol.~17, pp.~261--272, 2020.

\bibitem{dss-256}
{National Institute of Standards and Technology (NIST)}, ``Digital signature
  standard (dss) federal information processing standards publication 186-3.''

\bibitem{barenco1995elementarygates}
A.~Barenco, C.~H. Bennett, R.~Cleve, D.~P. DiVincenzo, N.~Margolus, P.~Shor,
  T.~Sleator, J.~A. Smolin, and H.~Weinfurter, ``Elementary gates for quantum
  computation,'' {\em Phys. Rev. A}, vol.~52, pp.~3457--3467, Nov 1995.

\end{thebibliography}

% ===start: IEEE===
%\EOD
% ===end: IEEE===

\end{document}